\documentclass[12pt]{article}
\usepackage{amsmath}
\usepackage{graphicx,psfrag,epsf}
\usepackage{enumerate}
\usepackage{amssymb}
\usepackage{booktabs}
\usepackage{array}
\usepackage{multirow}
\usepackage{booktabs}
\newcommand{\blind}{0}

\usepackage[authoryear]{natbib}
\usepackage[colorlinks=true, linkcolor=blue, citecolor=blue, urlcolor=blue]{hyperref}
\usepackage[skip=0.3cm]{caption}   
\usepackage{subcaption}   
\usepackage{algorithm}   
\usepackage{algpseudocode} 
\usepackage{amsthm}
\usepackage{float}
\newtheorem{theorem}{Theorem}[section]

\newtheorem{lemma}{Lemma}[section]

\newtheorem{definition}{Definition}[section]
\newtheorem{assumption}{Assumption}[section]

\begin{document}

\def\spacingset#1{\renewcommand{\baselinestretch}%
{#1}\small\normalsize} \spacingset{1}


\if0\blind
{
  \title{\bf MCMC Methods for Parameter Inference in Structurally Nonidentifiable Models}

  \author{
    Xuyuan Wang\thanks{Email: \texttt{xuyuan@ualberta.ca}},
    Donglin Han\thanks{Email: \texttt{donglin3@ualberta.ca}}
    and Michael Y. Li\thanks{Email: \texttt{myli@ualberta.ca}}\\
    Department of Mathematical and Statistical Sciences\\
    University of Alberta
  }

  \maketitle
}
\fi

\if1\blind
{
  \bigskip
  \bigskip
  \bigskip
  \begin{center}
    {\LARGE\bf MCMC Methods for Parameter Inference in Structurally Nonidentifiable Models}
  \end{center}
  \medskip
}
\fi
\begin{abstract}
We consider the problem of parameter inference for ordinary differential equation (ODE) models with structural non-identifiability. Such models arise in a wide range of scientific fields, including control theory, systems biology, and public health. Structural non-identifiability occurs when distinct parameter values provide identical model outputs, resulting in lower-dimensional manifolds of observationally equivalent solutions in the parameter space. This poses challenges for Bayesian inference and Markov chain Monte Carlo (MCMC) methods, often leading to poor mixing and slow convergence.
We develop two MCMC methods that use information from structural identifiability analysis. The first, Identifiability-Aware Geometric MCMC, constructs proposals that move within and between non-identifiable manifolds. The second, Identifiability-Aware Pseudo-Marginal MCMC, performs inference on the space of identifiable parameter combinations and reconstructs full parameter values. We show that both methods target the correct posterior distribution and are ergodic under standard conditions.
Numerical examples demonstrate improved sampling efficiency and convergence compared with standard MCMC methods.
\end{abstract}
\noindent%

\spacingset{1.45}
\section{Introduction}
\label{sec:intro}

The use of parametric ordinary differential equation (ODE) models to track the evolution of dynamical systems has been widely adopted across a variety of scientific fields, including control theory \citep{J1}, systems biology \citep{J2}, and public health \citep{J3}. A central task in the modeling process is parameter inference and uncertainty quantification based on observed data. These inferred results are then used to estimate latent quantities, make predictions, and support decision-making.
With the advancement of computational resources, Bayesian inference and Markov chain Monte Carlo (MCMC) methods have become increasingly popular for such tasks. Their advantage
lies in their flexibility in modeling assumptions, their ability to provide direct quantification of uncertainty, and their capacity to sample from
general
posterior probability distributions \citep{J4,J5,J6}. While the theoretical foundations of these methods are well established, their practical performance can deteriorate when the underlying model is structurally non-identifiable. For example, it is known that the classical susceptible--infectious (SI) model
\begin{equation}\label{eq:SI}
\begin{cases}
\displaystyle \frac{dS}{dt} = -\beta SI + \gamma I,\\
\displaystyle \frac{dI}{dt} = \beta SI - \gamma I,
\end{cases}
\end{equation}
with a partially observed case trajectory $\rho I(t)$ is structurally non-identifiable when only initial population is known \citep{cunn}. This implies that different parameter values can generate equivalent case trajectory. As a result, standard MCMC methods become trapped along non-identifiable directions, leading to slow mixing and poor convergence, as shown in Figure~\ref{fig1}.

\begin{figure}[H]
    \centering

    \begin{subfigure}[t]{0.45\textwidth}
        \centering
        \includegraphics[width=\linewidth]{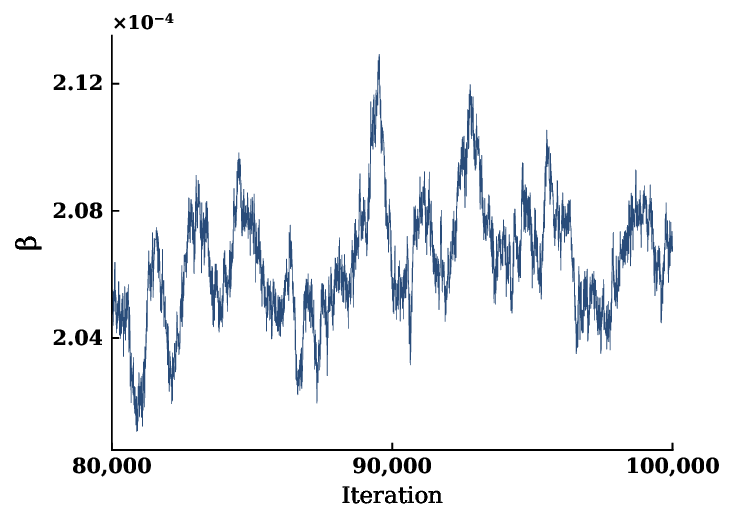}
        \subcaption{Trace plot for $\beta$.}
    \end{subfigure}\hfill
    \begin{subfigure}[t]{0.45\textwidth}
        \centering
        \includegraphics[width=\linewidth]{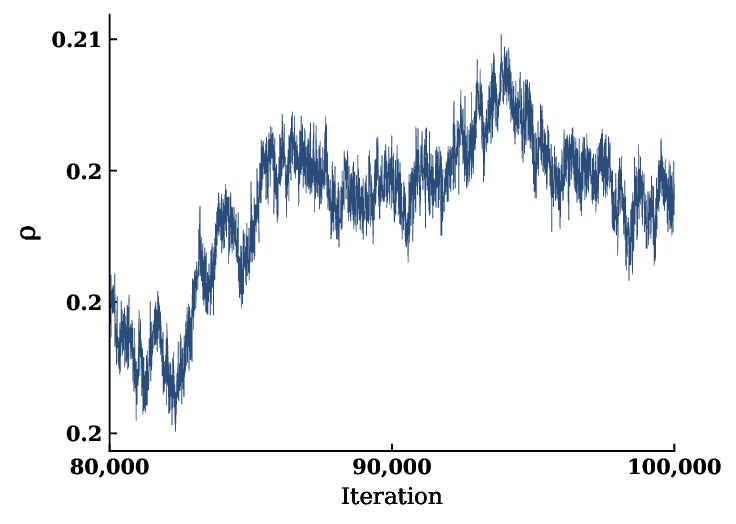}
        \subcaption{Trace plot for $\rho$.}
    \end{subfigure}
    \caption{Trace plots of the standard random-walk MCMC algorithm for the structurally non-identifiable SI model.}
    \label{fig1}
\end{figure}

Structural non-identifiability occurs when multiple parameter values produce identical observable model outputs. In such cases, the effect of changing one parameter can be compensated by adjusting others.
This phenomenon has been rigorously studied in the mathematical literature, where different definitions have been developed. For example, \citet{J7} develop approaches based on power series expansions; \citet{J8, J9} use generating series; \citet{J12} adopts a differential algebra framework; and \citet{J13} studies the problem from a differential geometric perspective. For a comprehensive review of these approaches, we refer readers to \citep{J14}.
Beyond theoretical developments, another substantial body of work has examined the impact of structural non-identifiability on parameter inference. Early work by \citet{J15} demonstrated that identifiability is fundamentally tied to the structure of the likelihood. Subsequently, \citet{J16, J17} developed a Bayesian framework for handling non-identifiability, emphasizing the role of prior distributions in resolving the issue. More recent studies have highlighted the numerical challenges induced by non-identifiability. For example, \citet{J19} show that optimization algorithms can be trapped along non-identifiable submanifolds, leading to slow or unstable convergence. Similarly, \citet{J20} showed that MCMC methods often suffer from poor mixing, as the sampler is constrained to move along sets of observationally equivalent parameter values.

Motivated by this literature, several methods have been proposed to address these challenges. A recent work by \citet{J25} introduced a general framework for constructing identifiability-aware samplers and demonstrated its effectiveness on simple economic models. Despite its promise, this framework has not yet been generalized to more complex rational ODE systems, for which the geometry of the identifiable parameter space is important.
Another line of research focuses on reducing the parameter space to a lower-dimensional identifiable subspace, where sampling can be performed more efficiently and the instability induced by non-identifiability is mitigated. These methods typically rely on identifying likelihood-informed linear subspaces that capture the dominant identifiable directions, followed by performing inference and optimization within these reduced spaces \citep{J21, J22, J23}.
A limitation of such methods is that structural non-identifiability is often nonlinear, whereas likelihood-informed subspaces provide only linear approximations of the identifiable structure. These approximations can introduce systematic bias into posterior inference. This limitation motivates the development of methods that explicitly incorporate structural identifiability analysis into Bayesian inference to reduce the parameter space to general identifiable manifolds.

In this paper, we develop two MCMC methods that leverage structural identifiability information. The first extends the identifiability-aware sampling framework of \citet{J25} to general rational ODE models. The resulting sampler incorporates a geometric proposal mechanism based on symplectic integrators, enabling efficient exploration along non-identifiable manifolds.
The second method is based on the pseudo-marginal MCMC framework. Rather than approximating the identifiable structure by a linear subspace, it directly exploits the nonlinear geometry induced by structural identifiability. Inference is performed on the lower-dimensional subspace of identifiable parameter combinations, and full parameter values are reconstructed conditionally on these combination values. This approach generalizes likelihood-informed subspace methods to nonlinear non-identifiable manifolds.
We show that both methods enable efficient posterior exploration in the presence of structural non-identifiability. The numerical instabilities associated with non-identifiable parameter directions are substantially reduced.
The main contributions of this work are as follows.
(i) We develop an identifiability-aware geometric MCMC method for rational ODE models that exploits the nonlinear manifold structure induced by structural non-identifiability.
(ii) We introduce a pseudo-marginal MCMC method that performs inference on identifiable parameter combinations and generalizes likelihood-informed subspace methods to nonlinear identifiable manifolds.
(iii) We provide theoretical convergence properties for the proposed sampling algorithms.
(iv) The proposed frameworks are broadly applicable to rational ODE models and can be employed across a wide range of scientific applications.
The remainder of the paper is organized as follows. Section~\ref{sec:pre} introduces the necessary preliminaries. Section~\ref{sec:meth} presents the proposed methodologies. Section~\ref{sec:case} contains case studies demonstrating the performance of the proposed methods. Section~\ref{sec:dis} concludes with a discussion.

\section{Preliminary} \label{sec:pre} 
We begin by reviewing basic concepts in parametric ODE models, structural non-identifiability, and Bayesian inference. Many widely used mathematical models can be expressed as systems of algebraic differential equations of the form
\begin{equation}\label{eq:ade}
\Sigma(\boldsymbol \theta) := 
\begin{cases}
\boldsymbol{x}'(t,\boldsymbol{\theta}) = \boldsymbol{f}(\boldsymbol{x}(t), \boldsymbol{\beta}), \\
\boldsymbol{y}(t,\boldsymbol{\theta}) = \boldsymbol{g}(\boldsymbol{x}(t), \boldsymbol{\beta}), \\
\boldsymbol{x}(0) = \boldsymbol{x}_0,
\end{cases}
\end{equation}
where $\boldsymbol{x}(t) \in \mathbb{R}^n$ denotes the state variables, $\boldsymbol{y}(t) \in \mathbb{R}^m$ the observable outputs, and $\boldsymbol{\beta} \in \mathbb{R}^p$ the unknown model parameters. The initial condition is given by $\boldsymbol{x}_0 \in \mathbb{R}^n$. The functions $\boldsymbol{f} = (f_1, \dots, f_n)$ and $\boldsymbol{g} = (g_1, \dots, g_m)$ are assumed to be rational functions.
To include cases where the initial conditions are also unknown, we define the full parameter vector as $\boldsymbol{\theta} := (\boldsymbol{\beta}, \boldsymbol{x}_0) \in \Theta \subset \mathbb{R}^{p+n}$, where $\Theta$ is a connected subset of Euclidean space. The system $\Sigma$ can be reformulated so that the initial conditions are treated as part of the parameter vector, resulting in a system with known initial conditions.
The calibration task is to infer $\boldsymbol{\theta}$ from observations of the output $\boldsymbol{y}$ at a discrete set of time points.
Structural identifiability addresses whether the parameter vector $\boldsymbol{\theta}$ can be uniquely recovered from ideal (noise-free and continuous) observations of the output. We now give a formal definition following \citet{J26} and  \citet{J27}.
\begin{definition}\label{def:noniden}
A parameter $\theta_i$, for $i \in \{1, 2, \dots, p+n\}$, is said to be structurally globally identifiable if, for almost every $\boldsymbol{\theta}^* \in \Theta$,
\[
\Sigma(\boldsymbol{\theta}) = \Sigma(\boldsymbol{\theta}^*) \;\Rightarrow\; \boldsymbol{\theta} = \boldsymbol{\theta}^*.
\]
A parameter $\theta_i$ is said to be structurally non-identifiable if, for almost every $\boldsymbol{\theta}^* \in \Theta$, there does not exist a neighborhood $V(\boldsymbol{\theta}^*)$ such that, for all $\boldsymbol{\theta} \in V(\boldsymbol{\theta}^*)$,
\[
\Sigma(\boldsymbol{\theta}) = \Sigma(\boldsymbol{\theta}^*) \;\Rightarrow\; \boldsymbol{\theta} = \boldsymbol{\theta}^*.
\]
\end{definition}
Based on Definition~\ref{def:noniden}, many methods have been developed to detect and assess structural non-identifiability in systems of the form $\Sigma$. A common approach is based on Lie group theory, which derives a system of input--output equations whose solvability properties provide information about structural identifiability \citep{J26}. Several software tools implementing this approach have also been developed \citep{J28,J29}.
We refer to this procedure as \emph{structural identifiability analysis}. It gives a rational mapping of identifiable combinations
\[
\boldsymbol{\xi} : \Theta \longrightarrow \mathcal{C} \subset \mathbb{R}^q,
\]
which maps the full parameter vector $\boldsymbol{\theta} \in \Theta \subset \mathbb{R}^{p+n}$ to a set of $q \leq p+n$ structurally identifiable combinations $\boldsymbol{c}\in\mathcal{C}\subset\mathbb{R}^q$. These combinations are obtained as coefficients in the input-output equations.
As a result, parameter values that lie in the same level set
\[
\mathcal{M}_{\boldsymbol{c}} := \left\{ \boldsymbol{\theta} \in \Theta : \boldsymbol{\xi}(\boldsymbol{\theta}) = \boldsymbol{c} \right\}
\]
produce identical observable outputs, i.e.,
\[
\boldsymbol{y}(t,\boldsymbol{\theta}) = \boldsymbol{y}(t, \boldsymbol{\theta}'), \quad \forall \, \boldsymbol{\theta}, \boldsymbol{\theta}' \in \mathcal{M}_{\boldsymbol{c}}.
\]
Under the standard Bayesian framework, this implies that the likelihood function is constant on each level set $\mathcal{M}_{\boldsymbol{c}}$, and posterior inference along these sets is determined entirely by the prior distribution over $\Theta$.

We now formalize the parameter inference problem within the Bayesian inverse problem framework. Given a model of the form $\Sigma$, let the prior distribution on the parameter space $\Theta \subseteq \mathbb{R}^{p+n}$ be specified by the density
\begin{equation}\label{eq:prior}
    p(\boldsymbol{\theta})
\propto
\mathbf{1}_{\boldsymbol{\theta}\in\Theta}\, p_0(\boldsymbol{\theta}),
\end{equation}
where $p_0$ is a known density on $\mathbb{R}^{p+n}$, and the indicator function restricts the prior support to the admissible parameter space $\Theta$. Suppose observations are collected at discrete time points,
\[
D := \{(t_1, \boldsymbol{y}_1), (t_2, \boldsymbol{y}_2), \dots, (t_T, \boldsymbol{y}_T)\}.
\]
Assuming a noise model with negative log-likelihood function $\rho(\cdot)$, the likelihood can be written as
\[
L(\boldsymbol{\theta}; D) \propto \exp\left(- \sum_{i=1}^T \rho\big(\boldsymbol{y}_i - \boldsymbol{y}(t_i; \boldsymbol{\theta})\big)\right).
\]
The posterior distribution is then given by
\[
\pi(\boldsymbol{\theta} \mid D) = \frac{L(\boldsymbol{\theta}; D)\, p(\boldsymbol{\theta})}{Z},
\]
where $Z$ is a normalizing constant chosen such that $\int_\Theta \pi(\boldsymbol{\theta} \mid D)d\boldsymbol{\theta} = 1$.
In general, the posterior density $\pi(\boldsymbol{\theta} \mid D)$ does not admit a closed-form expression and must be explored using sampling based methods such as Markov chain Monte Carlo (MCMC). In the presence of structural non-identifiability, the posterior is supported along manifolds of the form $\mathcal{M}_{\boldsymbol{c}}$. The inference problem becomes ill-posed when the prior is weakly informative, as an entire manifold of parameter values yields identical posterior. This often leads to poor mixing and slow convergence of standard MCMC algorithms, as we shown in Figure~\ref{fig1}. In Section~\ref{sec:meth}, we propose two classes of MCMC algorithms that leverage the results of structural identifiability analysis to improve sampling efficiency and convergence.
\section{Method} \label{sec:meth} 
We begin by considering the differentiability of the mapping $\boldsymbol{\xi}$. From structural identifiability analysis, $\boldsymbol{\xi}$ admits an analytic rational representation
\begin{equation}\label{eq:xi}
\boldsymbol{\xi}(\boldsymbol{\theta}) = \frac{\boldsymbol{P}(\boldsymbol{\theta})}{\boldsymbol{Q}(\boldsymbol{\theta})}, \quad \boldsymbol{\theta}\in \Theta \setminus \{\boldsymbol{\theta} : \boldsymbol{Q}(\boldsymbol{\theta}) = 0\},
\end{equation}
where $\boldsymbol{P}$ and $\boldsymbol{Q}$ are multivariate polynomials, and $\boldsymbol{\xi}$ is smooth on $\Theta \setminus \{\boldsymbol{\theta} : \boldsymbol{Q}(\boldsymbol{\theta}) = 0\}$.
For $\boldsymbol{\theta} \in \Theta$, the Jacobian matrix is
\begin{equation}\label{eq:dxi}
D\boldsymbol{\xi}(\boldsymbol{\theta})
=
\frac{D\boldsymbol{P}(\boldsymbol{\theta})\,\boldsymbol{Q}(\boldsymbol{\theta})
-
\boldsymbol{P}(\boldsymbol{\theta})\,D\boldsymbol{Q}(\boldsymbol{\theta})}
{\boldsymbol{Q}(\boldsymbol{\theta})^2}
\;\in\; \mathbb{R}^{q \times (p+n)}.
\end{equation}
By standard results in algebraic geometry, the set on which $\operatorname{rank}(D\boldsymbol{\xi}(\boldsymbol{\theta})) < q$ is contained in a algebraic variety and hence has Lebesgue measure zero if
\(
\boldsymbol{\xi}
\)
has generic rank \(q\) \citep{J24}. Therefore, $\boldsymbol{\xi}$ is a submersion almost everywhere, and the level sets
\(
\mathcal{M}_{\boldsymbol{c}} \)
are $(p+n-q)$-dimensional embedded submanifolds for almost every $\boldsymbol{c}\subset \mathcal{C}$. This property provides the foundation for derivative based geometric MCMC methods targeting the algebraic manifold induced by $\boldsymbol{\xi}$.
Let $\pi(\boldsymbol{\theta})$ be a posterior density on $\Theta$. By the coarea formula, for any integrable function $f$,
\[
\int_{\Theta} f(\boldsymbol{\theta}) \, d\boldsymbol{\theta}
=
\int_{\mathcal{C}} 
\left(
\int_{\mathcal{M}_{\boldsymbol{c}}}
\frac{f(\boldsymbol{\theta})}{J_{\boldsymbol{\xi}}(\boldsymbol{\theta})}
\, d\mathcal{H}^{p+n-q}(\boldsymbol{\theta})
\right)
d\boldsymbol{c},
\]
where $J_{\boldsymbol{\xi}}(\boldsymbol{\theta}) := \sqrt{\det\left( D\boldsymbol{\xi}(\boldsymbol{\theta}) D\boldsymbol{\xi}(\boldsymbol{\theta})^\top \right)}$ is the Jacobian determinant of the mapping $\boldsymbol{\xi}$, and $\mathcal{H}^{p+n-q}$ denotes the $(p+n-q)$-dimensional Hausdorff measure on the manifold $\mathcal{M}_{\boldsymbol{c}}$. 
Applying this to $f=\pi$ gives the marginal density on the identifiable subspace $\mathcal{C}\subset \mathbb R^q$
\begin{equation}\label{eq:pic}
\pi_{\mathcal{C}}(\boldsymbol{c})
=
\int_{\mathcal{M}_{\boldsymbol{c}}}
\frac{\pi(\boldsymbol{\theta})}{J_{\boldsymbol{\xi}}(\boldsymbol{\theta})}
\, d\mathcal{H}^{p+n-q}(\boldsymbol{\theta}), 
\end{equation}
where $\boldsymbol{c} = \boldsymbol{\xi}(\boldsymbol{\theta})$.
The conditional distribution of $\boldsymbol{\theta}$ given $\boldsymbol{c}$ is supported on $\mathcal{M}_{\boldsymbol{c}}$ and admits the density with respect to the Hausdorff measure $\mathcal{H}^{p+n-q}$
\begin{equation}\label{eq:picond}
\pi(\boldsymbol{\theta} \mid \boldsymbol{c})
=
\frac{
\displaystyle \frac{\pi(\boldsymbol{\theta})}{J_{\boldsymbol{\xi}}(\boldsymbol{\theta})}
}{
\displaystyle \int_{\mathcal{M}_{\boldsymbol{c}}}
\frac{\pi(\boldsymbol{\theta}')}{J_{\boldsymbol{\xi}}(\boldsymbol{\theta}')}
\, d\mathcal{H}^{p+n-q}(\boldsymbol{\theta}')
},
\quad \boldsymbol{\theta} \in \mathcal{M}_{\boldsymbol{c}}.
\end{equation}
Overall, the posterior admits the decomposition
\begin{equation}\label{eq:decom}
\pi(\boldsymbol{\theta})
=
\pi_{\mathcal{C}}(\boldsymbol{c})\,\pi(\boldsymbol{\theta}\mid\boldsymbol{c}).
\end{equation}
The decomposition \eqref{eq:decom} suggests two classes of sampling strategies for structurally non-identifiable models. Both exploit the geometric structure induced by the identifiable mapping \(\boldsymbol{\xi}\), but differ in how inference is performed over the manifold family \(\{\mathcal{M}_{\boldsymbol c}\}_{\boldsymbol c\in\mathcal C}\).
A first class of methods constructs geometric MCMC algorithms directly on the full parameter space \(\Theta\) \citep{J25}. The sampler employs a two-step proposal mechanism consisting of a \emph{teleportation move}, which proposes observationally equivalent states along the non-identifiable manifold \(\mathcal{M}_{\boldsymbol c}\), and a \emph{transition move}, which proposes states transverse to \(\mathcal{M}_{\boldsymbol c}\) in order to explore nearby manifolds. The composition of these two proposal steps defines a single proposal update, followed by a Metropolis--Hastings accept--reject correction targeting the posterior distribution \(\pi\).
Relative to standard random-walk or Hamiltonian samplers, this geometry-aware construction can improve mixing and mitigate slow exploration caused by posterior degeneracies arising from structural non-identifiability.
In Section~\ref{sec:geo}, we discuss how this strategy can be implemented using constrained Hamiltonian Monte Carlo methods.

A second approach performs inference on the lower-dimensional identifiable space \(\mathcal{C}\). Instead of sampling in the full parameter space \(\Theta\), the procedure proceeds in two steps:
\begin{enumerate}
    \item sample \(\boldsymbol{c} \sim \pi_{\mathcal{C}}(\boldsymbol{c})\) in the identifiable space \(\mathcal{C}\);
    
    \item conditional on \(\boldsymbol{c}\), reconstruct \(\boldsymbol{\theta} \sim \pi(\boldsymbol{\theta}\mid\boldsymbol{c})\) on the manifold \(\mathcal{M}_{\boldsymbol{c}}\).
\end{enumerate}
This decomposition separates identifiable and non-identifiable directions, and reduce the effective dimensionality of the inference problem.
In practice, direct evaluation of the marginal density \(\pi_{\mathcal{C}}\) in \eqref{eq:pic} is generally impossible due to the complex geometry of the algebraic manifold \(\mathcal{M}_{\boldsymbol{c}}\). The defining integral involves integration with respect to the Hausdorff measure over a manifold, which is difficult to evaluate analytically or numerically.
Nevertheless, the integrand in \eqref{eq:pic} can be evaluated almost everywhere up to a normalizing constant. This motivates a pseudo-marginal approach \citep{J30}, which we discussion in Section \ref{sec3.1}. 

\subsection{Geometric MCMC}\label{sec:geo}

We first discuss geometric MCMC methods for structurally non-identifiable
models. Given the identifiable combination map \(\boldsymbol{\xi}\), we
construct two proposal components: a transition move that explores nearby
identifiable level sets, and a teleportation move that evolves along the
non-identifiable manifold \(\mathcal{M}_{\boldsymbol c}\). The transition proposal can be chosen from a broad class of standard MCMC
updates, including random-walk Metropolis, Gibbs sampling, or Hamiltonian
Monte Carlo. Since the identifiable mapping
\(\boldsymbol{\xi}\) characterizes directions of structural
non-identifiability, isotropic proposals may explore the parameter space
inefficiently. This motivates the construction of anisotropic proposals
that preferentially move in directions normal to
\(\mathcal{M}_{\boldsymbol c}\) to improve exploration across
nearby identifiable manifolds.
Let \(\boldsymbol{\theta}\in\Theta\) and
\(\boldsymbol c=\boldsymbol{\xi}(\boldsymbol{\theta})\). For a generic point with full rank
\(D\boldsymbol{\xi}(\boldsymbol{\theta})\), the tangent and
normal spaces at \(\boldsymbol{\theta}\) are defined by
\[
  T\mathcal{M}_{\boldsymbol c}(\boldsymbol{\theta})
  :=
  \ker\!\bigl(D\boldsymbol{\xi}(\boldsymbol{\theta})\bigr),
  \qquad
  N\mathcal{M}_{\boldsymbol c}(\boldsymbol{\theta})
  :=
  \operatorname{span}\!\bigl(
    D\boldsymbol{\xi}(\boldsymbol{\theta})^\top
  \bigr).
\]
The ambient space admits the orthogonal decomposition
\(
\mathbb{R}^{p+n}
:=
T\mathcal{M}_{\boldsymbol c}(\boldsymbol{\theta})
\oplus
N\mathcal{M}_{\boldsymbol c}(\boldsymbol{\theta}).
\)
The orthogonal projection onto the normal space is given by
\begin{equation}\label{eq:pin}
  \Pi_N(\boldsymbol{\theta})
  =
  D\boldsymbol{\xi}(\boldsymbol{\theta})^\top
  \left(
    D\boldsymbol{\xi}(\boldsymbol{\theta})
    D\boldsymbol{\xi}(\boldsymbol{\theta})^\top
  \right)^{-1}
  D\boldsymbol{\xi}(\boldsymbol{\theta}).
\end{equation}
We define the transition proposal density by
\begin{equation}\label{eq:qnormal}
  q_P(\boldsymbol{\theta}'\mid\boldsymbol{\theta})
  =
  \mathcal{N}\!\left(
    \boldsymbol{\theta}';\,
    \boldsymbol{\theta},\,
    \sigma_N^2\,\Pi_N(\boldsymbol{\theta})
    +
    \sigma_T^2\bigl(I-\Pi_N(\boldsymbol{\theta})\bigr)
  \right),
\end{equation}
where \(\sigma_N>\sigma_T>0\). This construction permits larger moves in
identifiable directions. We will show next that although the transition proposal \(q_P\) is
anisotropic, the resulting Metropolis-Hastings acceptance rate is analytical when combined with a reversible and
volume-preserving teleportation map.

The teleportation proposal evolves the proposed state along
\(\mathcal{M}_{\boldsymbol c}\). A central challenge is that
\(\mathcal{M}_{\boldsymbol c}\) is an implicitly defined algebraic
manifold, so proposal mechanisms are
generally not available in closed form. Constrained Hamiltonian
integrators such as RATTLE \citep{J36} provide a way for
constructing such proposals.
The method first samples a velocity variable
\(
p\in T\mathcal{M}_{\boldsymbol c}(\boldsymbol{\theta}).
\)
Then a step size \(\varepsilon>0\) is
used to move the state in the tangent direction, followed by a
projection step that maps the updated state back onto the manifold through
Lagrange multiplier corrections. The resulting trajectory 
remains constrained to \(\mathcal{M}_{\boldsymbol c}\).
A  single RATTLE update is summarized in
Algorithm~\ref{alg:rattle}.
\begin{algorithm}[H]
\caption{RATTLE teleportation}
\label{alg:rattle}
\begin{algorithmic}[1]

\Require Current state
\(
\boldsymbol{\theta}_n\in\mathcal{M}_{\boldsymbol c}
\),
step size \(\varepsilon>0\)

\State Sample a tangent velocity
\(
p_n\in T\mathcal{M}_{\boldsymbol c}(\boldsymbol{\theta}_n)
\)
such that
\(
D\boldsymbol{\xi}(\boldsymbol{\theta}_n)p_n=0.
\)

\State Perform a position update
\(
\boldsymbol{\theta}'
=
\boldsymbol{\theta}_n
+
\varepsilon p_n.
\)

\State Project
\(
\boldsymbol{\theta}'
\)
back onto the manifold by solving
\(
\boldsymbol{\xi}(\boldsymbol{\theta}_{n+1})
=
\boldsymbol c
\).

\State Project the updated velocity onto $T\mathcal{M}_{\boldsymbol{c}}(\boldsymbol{\theta}_{n+1})$ through $\Pi_N$ in \eqref{eq:pin}.

\State \Return
\(
(\boldsymbol{\theta}_{n+1},p_{n+1})
\).

\end{algorithmic}
\end{algorithm}
The result of a RATTLE step generates a new state on
\(\mathcal{M}_{\boldsymbol c}\). Importantly, the RATTLE integrator is reversible and
volume-preserving \citep{M1}. These properties make it suitable for
exploring distributions supported on implicitly defined manifolds, and
have led to the development of several constrained Monte Carlo methods;
see, for example, \citet{M2,M3,M4}. Iterating the RATTLE update \(M\)
times defines a teleportation map
\[
T_M:\mathcal{M}_{\boldsymbol c}\to\mathcal{M}_{\boldsymbol c},
\]
which generates long-range moves along the non-identifiable manifold.
Let
\[
P(\boldsymbol{\theta},d\widetilde{\boldsymbol{\theta}})
=
q_P(\widetilde{\boldsymbol{\theta}}\mid\boldsymbol{\theta})
\,d\widetilde{\boldsymbol{\theta}}
\]
denote the transition proposal kernel defined in
\eqref{eq:qnormal}. The full geometric proposal is obtained by composing
the transition move with the teleportation map
\[
\widetilde{\boldsymbol{\theta}}
\sim
P(\boldsymbol{\theta},\cdot),
\qquad
\boldsymbol{\theta}'
=
T_M(\widetilde{\boldsymbol{\theta}}).
\]
The full composite proposal kernel is then given by
\begin{equation}\label{eq:Q}
    Q(\boldsymbol{\theta},d\boldsymbol{\theta}')
=
P\bigl(
\boldsymbol{\theta},
T_M^{-1}(d\boldsymbol{\theta}')
\bigr).
\end{equation}
The composite kernel \(Q\) does not preserve the target distribution \(\pi\) unless a Metropolis--Hastings correction is applied. The correction ensures that the resulting Markov chain admits \(\pi\) as its invariant distribution. A key requirement is that the corrected kernel satisfies detailed balance with respect to \(\pi\). In Lemma~\ref{lem:mhgeo}, we derive the Metropolis--Hastings acceptance probability for $Q$ and show that the resulting Markov kernel satisfies detailed balance with respect to \(\pi\).

\begin{lemma}\label{lem:mhgeo}
The geometric
MCMC proposal $Q$ satisfies detailed balance with respect to
\(\pi\) under the Metropolis--Hastings acceptance probability
\begin{equation}\label{eq:mhgeo}
\alpha(\boldsymbol{\theta},\boldsymbol{\theta}')
=
1\wedge
\frac{
\pi(\boldsymbol{\theta}')
q_P(\widetilde{\boldsymbol{\theta}}\mid\boldsymbol{\theta}')
}{
\pi(\boldsymbol{\theta})
q_P(\widetilde{\boldsymbol{\theta}}'\mid\boldsymbol{\theta})
},
\end{equation}
where
\[
\widetilde{\boldsymbol{\theta}}
=
T_M^{-1}(\boldsymbol{\theta}'),
\qquad
\widetilde{\boldsymbol{\theta}}'
=
T_M^{-1}(\boldsymbol{\theta}).
\]
\end{lemma}

\begin{proof}
A single RATTLE step is the composition of symplectic drift, constraint
projection, and momentum projection maps on the constrained phase space
\[
\mathcal D
:=
\left\{
(\boldsymbol{\theta},p):
\boldsymbol{\xi}(\boldsymbol{\theta})=\boldsymbol c,
\;
D\boldsymbol{\xi}(\boldsymbol{\theta})p=0
\right\}.
\]
Each component map is symplectic, hence preserves the canonical volume
form. Therefore the RATTLE map
\(
T:\mathcal D\to\mathcal D
\)
satisfies
\(
T^*\omega=\omega,
\)
where \(\omega\) is the constrained symplectic form. By Liouville's
theorem,
\[
\left|\det DT(\boldsymbol{\theta},p)\right|=1.
\]
Since
\(
T_M=T^M
\), we have
\[
\left|
\det DT_M(\boldsymbol{\theta},p)
\right|
=
\prod_{k=0}^{M-1}
\left|
\det DT(T^k(\boldsymbol{\theta},p))
\right|
=
1.
\] 
Hence \(T_M\) is volume-preserving. The RATTLE step is also symmetric under momentum reversal. If
\(
(\boldsymbol{\theta}_n,p_n)
\mapsto
(\boldsymbol{\theta}_{n+1},p_{n+1})
\)
is one RATTLE step, then applying the same update with step size
\(\varepsilon\) and reversed momentum recovers the original state
\[
(\boldsymbol{\theta}_{n+1},-p_{n+1})
\mapsto
(\boldsymbol{\theta}_{n},-p_n).
\]
Hence the map is time-reversible. Since compositions of reversible maps
remain reversible, the \(M\)-step teleportation map \(T_M\) is
reversible on \(\mathcal{M}_{\boldsymbol c}\).
Let
\[
\widetilde{\boldsymbol{\theta}}
=
T_M^{-1}(\boldsymbol{\theta}'),
\qquad
\widetilde{\boldsymbol{\theta}}'
=
T_M^{-1}(\boldsymbol{\theta}).
\]
The proposal mechanism first samples
\(
\widetilde{\boldsymbol{\theta}}
\sim
q_P(\cdot\mid\boldsymbol{\theta}),
\)
then deterministically maps the proposal through \(T_M\) by evaluating
\(
\boldsymbol{\theta}'
=
T_M(\widetilde{\boldsymbol{\theta}}).
\)
Therefore, the proposal kernel is the pushforward of \(q_P\) through
\(T_M\). By the change-of-variables formula,
\[
Q(\boldsymbol{\theta},d\boldsymbol{\theta}')
=
q_P(
T_M^{-1}(\boldsymbol{\theta}')
\mid
\boldsymbol{\theta}
)
\left|
\det DT_M^{-1}(\boldsymbol{\theta}')
\right|
d\boldsymbol{\theta}'.
\]
Since \(T_M\) is volume-preserving, we have
\(
\left|
\det DT_M^{-1}(\boldsymbol{\theta}')
\right|
=
1,
\)
and hence
\[
Q(\boldsymbol{\theta},d\boldsymbol{\theta}')
=
q_P(
T_M^{-1}(\boldsymbol{\theta}')
\mid
\boldsymbol{\theta}
)
\,d\boldsymbol{\theta}'.
\]
Reversibility of \(T_M\) implies that every forward trajectory from
\(
\widetilde{\boldsymbol{\theta}}
\)
to
\(
\boldsymbol{\theta}'
\)
is paired with a unique reverse trajectory from
\(
\widetilde{\boldsymbol{\theta}}'
\)
to
\(
\boldsymbol{\theta}
\).
Therefore the reverse proposal density is
\[
Q(\boldsymbol{\theta}',d\boldsymbol{\theta})
=
q_P(
T_M^{-1}(\boldsymbol{\theta})
\mid
\boldsymbol{\theta}'
)
\,d\boldsymbol{\theta}.
\]
The Metropolis--Hastings ratio then becomes
\[
\frac{
\pi(\boldsymbol{\theta}')
Q(\boldsymbol{\theta}',d\boldsymbol{\theta})
}{
\pi(\boldsymbol{\theta})
Q(\boldsymbol{\theta},d\boldsymbol{\theta}')
}
=
\frac{
\pi(\boldsymbol{\theta}')
q_P(
T_M^{-1}(\boldsymbol{\theta})
\mid
\boldsymbol{\theta}'
)
}{
\pi(\boldsymbol{\theta})
q_P(
T_M^{-1}(\boldsymbol{\theta}')
\mid
\boldsymbol{\theta}
)
},
\]
which gives \eqref{eq:mhgeo}.
Define the Markov kernel
\begin{equation}\label{eq:K}
    K(\boldsymbol{\theta},d\boldsymbol{\theta}')
:=
Q(\boldsymbol{\theta},d\boldsymbol{\theta}')
\alpha(\boldsymbol{\theta},\boldsymbol{\theta}')
+
r(\boldsymbol{\theta})
\delta_{\boldsymbol{\theta}}
(d\boldsymbol{\theta}'),
\end{equation}
where
\(
r(\boldsymbol{\theta})
=
1-
\int
Q(\boldsymbol{\theta},d\boldsymbol{\theta}')
\alpha(\boldsymbol{\theta},\boldsymbol{\theta}').
\)
By the standard Metropolis--Hastings construction,
\[
\pi(\boldsymbol{\theta})
K(\boldsymbol{\theta},d\boldsymbol{\theta}')
=
\pi(\boldsymbol{\theta}')
K(\boldsymbol{\theta}',d\boldsymbol{\theta}),
\]
and \(K\) satisfies detailed balance with respect to \(\pi\).
\end{proof}
From the construction of \(T_M\), each RATTLE update requires solving the algebraic system
\(
\boldsymbol{\xi}(\boldsymbol{\theta}_{n+1})=\boldsymbol c
\), and an algebraic equation solver is needed. For the present setting, a local solver such as Newton's method is sufficient because the step size \(\varepsilon\) is chosen to be small. The solver converges to the desired nearby solution on the same manifold branch.
If the equation solve fails, the proposal is rejected. In practice, this nonlinear solve is incorporated into the \emph{reversibility check} of the RATTLE integrator. The reversibility check verifies that the numerical trajectory can be retraced by reversing the momentum and applying the integrator backward. This is important because reversibility is a key requirement for the Metropolis--Hastings correction and guarantees that the resulting Markov kernel satisfies detailed balance. A detailed definition and analysis of the reversibility check for constrained Hamiltonian dynamics can be found in \citet{M3}.

The teleportation map \(T_M\) generates updates on \(\mathcal M_{\boldsymbol c}\) through a sequence of small constrained position projections. A single application of \(T_M\) typically remains on the same connected component of the manifold. This raises the question of whether the algorithm remains ergodic when \(\mathcal M_{\boldsymbol c}\) is disconnected. Fortunately, ergodicity is inherited from the transition kernel $P$, provided that $P$ is itself ergodic with respect to the target distribution. The teleportation step serves only to improve exploration along non-identifiable directions.
The resulting identifiability-aware geometric MCMC algorithm is summarized in Algorithm~\ref{alg:IAGMCMC}. We next present a convergence theorem for the resulting Markov chain.
\begin{theorem}\label{thm:geoergodic}
Assume that \(\Theta\subset \mathbb R^{p+n}\) is compact and that the target density
\(\pi\) is continuous and strictly positive on \(\Theta\).
Then the Markov kernel $K$ in equation \eqref{eq:K} admits \(\pi\) as its unique invariant
distribution. Moreover, there exist constants \(C<\infty\) and
\(\rho\in(0,1)\) such that
\[
\sup_{\theta\in\Theta}
\left|
\delta_{\boldsymbol{\theta}} K^n-\pi
\right|_{\mathrm{TV}}
\le
C\rho^n,
\qquad n\ge 1.
\]
Hence \(K\) is uniformly geometrically ergodic.
\end{theorem}
\begin{proof}
By Lemma~\ref{lem:mhgeo}, the kernel \(K\) satisfies detailed balance,
\[
\pi(d\boldsymbol{\theta})K(\boldsymbol{\theta},d\boldsymbol{\theta}')
=
\pi(d\boldsymbol{\theta}')K(\boldsymbol{\theta}',d\boldsymbol{\theta}).
\]
Integrating both sides with respect to \(\boldsymbol{\theta}\) gives
\[
\int_\Theta
\pi(d\boldsymbol{\theta})\,
K(\boldsymbol{\theta},A)
=
\pi(A),
\qquad
A\in\mathcal B(\Theta),
\]
where \(\mathcal B(\Theta)\) denotes the Borel \(\sigma\)-algebra on
\(\Theta\). Therefore \(\pi\) is an invariant distribution of \(K\).
Since \(\Theta\) is compact and \(q_P\) is continuous and strictly
positive, the extreme value theorem implies that \(q_P\) attains its
minimum and maximum on \(\Theta\times\Theta\). Hence
\[
0<q_{\min}
:=
\inf_{\boldsymbol{\theta},\boldsymbol{\theta}'\in\Theta}
q_P(\boldsymbol{\theta}'\mid\boldsymbol{\theta})
\le
q_P(\boldsymbol{\theta}'\mid\boldsymbol{\theta})
\le
\sup_{\boldsymbol{\theta},\boldsymbol{\theta}'\in\Theta}
q_P(\boldsymbol{\theta}'\mid\boldsymbol{\theta})
=:q_{\max}
<\infty.
\]
Similarly, continuity and strict positivity of \(\pi\) on the compact
set \(\Theta\) imply
\[
0<\pi_{\min}
\le
\pi(\boldsymbol{\theta})
\le
\pi_{\max}
<\infty.
\]
For every \(\boldsymbol{\theta},\boldsymbol{\theta}'\in\Theta\),
\[
\alpha(\boldsymbol{\theta},\boldsymbol{\theta}')
=
1\wedge
\frac{
\pi(\boldsymbol{\theta}')
q_P(T_M^{-1}(\boldsymbol{\theta})\mid\boldsymbol{\theta}')
}{
\pi(\boldsymbol{\theta})
q_P(T_M^{-1}(\boldsymbol{\theta}')\mid\boldsymbol{\theta})
}.
\]
Using the bounds above,
\[
\frac{
\pi(\boldsymbol{\theta}')
q_P(T_M^{-1}(\boldsymbol{\theta})\mid\boldsymbol{\theta}')
}{
\pi(\boldsymbol{\theta})
q_P(T_M^{-1}(\boldsymbol{\theta}')\mid\boldsymbol{\theta})
}
\ge
\frac{\pi_{\min}q_{\min}}
     {\pi_{\max}q_{\max}}.
\]
Therefore $\alpha$ is bounded below by
\[
a_0
:=
\min\!\left\{
1,
\frac{\pi_{\min}q_{\min}}
     {\pi_{\max}q_{\max}}
\right\}
>0.
\]
Since \(T_M\) is volume-preserving, Lemma~\ref{lem:mhgeo} gives
\[
Q(\boldsymbol{\theta},d\boldsymbol{\theta}')
=
q_P(T_M^{-1}(\boldsymbol{\theta}')\mid\boldsymbol{\theta})\,d\boldsymbol{\theta}',
\]
where \(d\boldsymbol{\theta}'\) denotes Lebesgue measure on \(\Theta\). Hence
\[
Q(\boldsymbol{\theta},d\boldsymbol{\theta}')
\ge
q_{\min}\,d\boldsymbol{\theta}'.
\]
Combining this with the lower bound on the acceptance probability gives
\[
K(\boldsymbol{\theta},d\boldsymbol{\theta}')
=
Q(\boldsymbol{\theta},d\boldsymbol{\theta}')
\alpha(\boldsymbol{\theta},\boldsymbol{\theta}')
+
r(\boldsymbol{\theta})\delta_{\boldsymbol{\theta}}(d\boldsymbol{\theta}')
\ge
a_0 q_{\min}\,d\boldsymbol{\theta}'.
\]
Let
\(
\nu(A)
=
\lambda(A)/\lambda(\Theta),
\)
where \(\lambda\) denotes Lebesgue measure on \(\Theta\). Then
\[
K(\boldsymbol{\theta},A)
\ge
\delta\,\nu(A),
\qquad
\delta
=
a_0q_{\min}\lambda(\Theta)
>0.
\]
Hence \(K\) satisfies a global Doeblin minorization condition.
By Doeblin's theorem, \(K\) is uniformly ergodic and admits a unique
invariant probability measure. Since \(\pi\) has already been shown to
be invariant, it follows that \(\pi\) is the unique invariant
distribution of \(K\). Moreover,
\[
\sup_{\boldsymbol{\theta}\in\Theta}
\bigl\|
\delta_{\boldsymbol{\theta}} K^n-\pi
\bigr\|_{\mathrm{TV}}
\le
(1-\delta)^n,
\qquad n\ge1.
\]
Therefore \(K\) converges geometrically fast to \(\pi\) in total
variation distance, uniformly over all initial states.
\end{proof}
From Theorem~\ref{thm:geoergodic}, the identifiability-aware geometric MCMC algorithm is geometrically ergodic for a broad class of Bayesian inference problems with compact parameter spaces and continuous positive posterior densities. The introduction of the teleportation step does not compromise the theoretical convergence results of standard MCMC algorithms.
The teleportation map \(T_M\) uses structural identifiability information to move efficiently along non-identifiable manifolds. As a result, the sampler can explore posterior regions more effectively by improving mixing and reducing autocorrelation. An important observation is that the identifiable combination map
\(
\boldsymbol{\xi}:\Theta\rightarrow\mathcal C
\)
maps the full parameter space into a lower-dimensional space of identifiable combinations. This geometric structure suggests that structural identifiability analysis can be used not only to construct more efficient samplers, but also to perform dimension reduction. This idea is closely related to active subspace methods \citep{J35}, which identify low-dimensional linear projections that capture the dominant variation of a target function. 
Motivated by this, we next discuss a dimension-reduced inference framework based on the identifiable combination map \(\boldsymbol{\xi}\). 

\subsection{Pseudo-marginal MCMC}\label{sec3.1}
We consider a pseudo-marginal MCMC scheme for sampling from the marginal density $\pi_{\mathcal C}$ defined in equation \eqref{eq:pic}. Suppose that for each $\boldsymbol{c} \in \mathcal{C}$, we can construct a nonnegative random variable $\widehat{\pi}_{\mathcal C}(\boldsymbol{c}, U)$, where $U$ denotes auxiliary randomness, such that
\[
\mathbb{E}_U\!\left[\widehat{\pi}_{\mathcal C}(\boldsymbol{c}, U)\right]
=
\pi_{\mathcal C}(\boldsymbol{c}).
\]
Given a proposal density $q(\boldsymbol{c}' \mid \boldsymbol{c})$, 
a Metropolis–Hastings algorithm can be implemented by replacing $\pi_{\mathcal C}$ with $\widehat{\pi}_{\mathcal C}$  and using the acceptance probability
\begin{equation}\nonumber
\alpha
=
\min\left\{
1,\;
\frac{\widehat{\pi}_{\mathcal C}(\boldsymbol{c}', U')}{\widehat{\pi}_{\mathcal C}(\boldsymbol{c}, U)}
\cdot
\frac{q(\boldsymbol{c} \mid \boldsymbol{c}')}{q(\boldsymbol{c}' \mid \boldsymbol{c})}
\right\}.
\end{equation}
It can be shown that this construction defines a Markov chain on the extended space $(\boldsymbol{c}, U)$ whose marginal stationary distribution in $\boldsymbol{c}$ is $\pi_\mathcal{C}$, provided the estimator is unbiased \citep{J30}. 
Let $\tilde{\pi}(\boldsymbol{\theta}) := L(\boldsymbol{\theta};D)\,p_0(\boldsymbol{\theta})$ denote the unnormalized posterior, and define the unnormalized marginal
\[
\tilde{\pi}_{\mathcal C}(\boldsymbol{c})
=
\int_{\mathcal{M}_{\boldsymbol{c}}}
\frac{\tilde{\pi}(\boldsymbol{\theta})}
{J_{\boldsymbol{\xi}}(\boldsymbol{\theta})}
\, d\mathcal{H}^{p+n-q}(\boldsymbol{\theta}).
\]
For each $\boldsymbol{c} \in \mathcal{C}$, let $q_{\boldsymbol{c}}(\boldsymbol{\theta})$ be a probability density on $\mathcal{M}_{\boldsymbol{c}}$ with respect to the Hausdorff measure $\mathcal{H}^{p+n-q}$, and suppose we can generate i.i.d.\ samples
\[
\boldsymbol{\theta}^{(1)},\ldots,\boldsymbol{\theta}^{(N)}
\;\stackrel{\text{i.i.d.}}{\sim}\;
q_{\boldsymbol{c}}(\boldsymbol{\theta}).
\]
Let $U := (\boldsymbol{\theta}^{(1)},\ldots,\boldsymbol{\theta}^{(N)})$, and define the estimator
\begin{equation}\label{eq:est}
\widehat{\tilde{\pi}}_{\mathcal C}(\boldsymbol{c}, U)
=
\frac{1}{N}\sum_{i=1}^N
\frac{\tilde{\pi}\!\left(\boldsymbol{\theta}^{(i)}\right)}
{J_{\boldsymbol{\xi}}\!\left(\boldsymbol{\theta}^{(i)}\right)\,
q_{\boldsymbol{c}}\!\left(\boldsymbol{\theta}^{(i)}\right)}.
\end{equation}
We conclude that the estimator is unbiased for the unnormalized marginal posterior $\tilde{\pi}_{\mathcal C}$ by a standard importance sampling argument, as summarized in the following lemma.
\begin{lemma}\label{lem:3.1}
For any $\boldsymbol{c} \in \mathcal{C}$, the estimator $\widehat{\tilde{\pi}}_{\mathcal C}(\boldsymbol{c}, U)$ defined in equation \eqref{eq:est} is unbiased, i.e.,
\[
\mathbb{E}_U\!\left[
\widehat{\tilde{\pi}}_{\mathcal C}(\boldsymbol{c}, U)
\right]
=
\tilde{\pi}_{\mathcal C}(\boldsymbol{c}).
\]
\end{lemma}
\begin{proof}
By independence and identical distribution of the samples,
\[
\mathbb{E}_U\!\left[
\widehat{\tilde{\pi}}_{\mathcal C}(\boldsymbol{c}, U)
\right]
=
\mathbb{E}_{\boldsymbol{\theta} \sim q_{\boldsymbol{c}}}
\left[
\frac{\tilde{\pi}(\boldsymbol{\theta})}
{J_{\boldsymbol{\xi}}(\boldsymbol{\theta})\,
q_{\boldsymbol{c}}(\boldsymbol{\theta})}
\right].
\]
Using the definition of $q_{\boldsymbol{c}}$ with respect to the Hausdorff measure,
\[
=
\int_{\mathcal{M}_{\boldsymbol{c}}}
\frac{\tilde{\pi}(\boldsymbol{\theta})}
{J_{\boldsymbol{\xi}}(\boldsymbol{\theta})\,
q_{\boldsymbol{c}}(\boldsymbol{\theta})}
\, q_{\boldsymbol{c}}(\boldsymbol{\theta})
\, d\mathcal{H}^{p+n-q}(\boldsymbol{\theta})
=
\tilde{\pi}_{\mathcal C}(\boldsymbol{c}),
\]
which completes the proof.
\end{proof}

\begin{algorithm}[htbp]
\caption{Identifiability-aware geometric MCMC}
\label{alg:IAGMCMC}
\begin{algorithmic}[1]
\Require Initial state $\boldsymbol{\theta}_0\in\Theta$, proposal density $q_P$, number of steps $M$, and Algorithm~\ref{alg:rattle}.

\For{$k=0,\ldots,K-1$}

\State Compute
$
\boldsymbol{c}_k
=
\boldsymbol{\xi}(\boldsymbol{\theta}_k).
$

\State Generate teleportation proposal
$
\widetilde{\boldsymbol{\theta}}
=
T_M(\boldsymbol{\theta}_k)
$
through Algorithm \ref{alg:rattle} on
$
\mathcal M_{\boldsymbol c_k}.
$

\If{\textit{reversibility check} fails}
    \State Set
    $
    \boldsymbol{\theta}^{(k+1)}
    =
    \boldsymbol{\theta}^{(k)}
    $
    and continue.
\EndIf

\State Draw
$
\boldsymbol{\theta}'
\sim
q_P(\widetilde{\boldsymbol{\theta}},\cdot).
$

\State Compute the acceptance probability $\alpha$ based on equation \eqref{eq:alpha}.
\State Sample
\(
u\sim\mathrm{Unif}(0,1)
\).

\If{$u<\alpha$}
\State Set
\(
\boldsymbol{\theta}_{k+1}
= \boldsymbol{\theta}'.
\)
\Else
\State Set
\(
\boldsymbol{\theta}_{k+1}
=
\boldsymbol{\theta}_k.
\)
\EndIf
\EndFor
\end{algorithmic}
\end{algorithm}

Since $\tilde{\pi}_{\mathcal C}(\boldsymbol{c}) = Z\,\pi_{\mathcal C}(\boldsymbol{c})$ for an unknown constant $Z$, the estimator $\widehat{\tilde{\pi}}_{\mathcal C}(\boldsymbol{c}, U)$ can be used within the pseudo-marginal scheme, as the constant $Z$ cancels in the Metropolis--Hastings ratio $\alpha$.
Replacing $\pi_{\mathcal{C}}$ by the estimator in equation \eqref{eq:est},
the acceptance probability becomes
\begin{equation}\label{eq:alpha}
    \widetilde{\alpha}=
\min\left\{
1,\;
\frac{\widehat{\tilde{\pi}}_{\mathcal{C}}(\boldsymbol{c}',U')}
{\widehat{\tilde{\pi}}_{\mathcal{C}}(\boldsymbol{c},U)}
\cdot
\frac{q(\boldsymbol{c}\mid \boldsymbol{c}')}
{q(\boldsymbol{c}'\mid \boldsymbol{c})}
\right\}.
\end{equation}
It remains to construct a suitable probability density $q_{\boldsymbol{c}}$ on $\mathcal{M}_{\boldsymbol{c}}$ such that sampling from $q_{\boldsymbol{c}}$ is efficient and its density can be evaluated in the estimator \eqref{eq:est}. We discuss some practical choices for $q_{\boldsymbol{c}}$ in Section \ref{sec:3.2} and \ref{sec:3.3}.

\subsubsection{Direct sampling from $\mathcal{M}_{\boldsymbol{c}}$}\label{sec:3.2}
We first consider a simple setting where the manifold $\mathcal{M}_{\boldsymbol{c}}$ has a tractable geometry. This arises when structural non-identifiability induces a linear identifiable mapping $\boldsymbol{\xi}$, so that $\mathcal{M}_{\boldsymbol{c}}$ is an affine subspace. If the parameter space $\Theta$ is a simple analytic set (e.g., a hyper-rectangle), then $\mathcal{M}_{\boldsymbol{c}} \cap \Theta$ admits an analytical characterization and can be sampled efficiently. Such situations arise in several models in systems biology and computational chemistry \citep{J19, J31}. In this case, a natural choice is to take $q_{\boldsymbol{c}}$ as the uniform density on $\mathcal{M}_{\boldsymbol{c}} \cap \Theta$ with respect to the Hausdorff measure, i.e.,
\[
q_{\boldsymbol{c}}(\boldsymbol{\theta})
=
\frac{1}{\mathrm{Vol}(\mathcal{M}_{\boldsymbol{c}} \cap \Theta)},
\quad \boldsymbol{\theta} \in \mathcal{M}_{\boldsymbol{c}} \cap \Theta,
\]
where $\mathrm{Vol}(\mathcal{M}_{\boldsymbol{c}} \cap \Theta) := \mathcal{H}^{p+n-q}(\mathcal{M}_{\boldsymbol{c}} \cap \Theta)$ is the $(p+n-q)$-dimensional volume. Sampling from $q_{\boldsymbol{c}}$ then reduces to drawing random points uniformly on the affine subspace restricted to $\Theta$. Under this choice, the estimator in equation \eqref{eq:est} becomes
\[
\widehat{\tilde{\pi}}_{\mathcal C}(\boldsymbol{c})
=
\frac{\mathrm{Vol}(\mathcal{M}_{\boldsymbol{c}} \cap \Theta)}{N}
\sum_{i=1}^N
\frac{\tilde{\pi}(\boldsymbol{\theta}^{(i)})}
{J_{\boldsymbol{\xi}}(\boldsymbol{\theta}^{(i)})},
\quad
\boldsymbol{\theta}^{(i)} \sim q_{\boldsymbol{c}},
\]
and is an unbiased estimator to $\tilde{\pi}_{\mathcal{C}}$
by Lemma \ref{lem:3.1}. For a general identifiable mapping, direct sampling is often unavailable due to the implicit definition of the algebraic manifold $\mathcal{M}_{\boldsymbol{c}}$. This motivates the development of indirect sampling strategies that operate in the ambient space and subsequently project onto $\mathcal{M}_{\boldsymbol{c}}$, resulting in a tractable density on the manifold.

\subsubsection{Indirect sample from $\mathcal{M}_{\boldsymbol{c}}$}\label{sec:3.3}
Given a generic point $\boldsymbol{c}\in\mathcal{C}$, our goal is to generate i.i.d. samples from the manifold $\mathcal{M}_{\boldsymbol{c}}$ in order to construct a pseudo-marginal MCMC estimator for $\tilde{\pi}_{\mathcal{C}}(\boldsymbol{c})$. Since $\boldsymbol{\xi}$ admits the rational representation in \eqref{eq:xi}, the level set condition
\(
\boldsymbol{\xi}(\boldsymbol{\theta})=\boldsymbol{c}
\)
is characterized by the polynomial system
\begin{equation}\label{eq:poly}
\mathcal{M}_{\boldsymbol{c}}
=
\left\{
\boldsymbol{\theta}\in\Theta:
\boldsymbol{F}_{\boldsymbol c}(\boldsymbol{\theta})
:=
\boldsymbol{P}(\boldsymbol{\theta})
-
\boldsymbol{c}\,
\boldsymbol{Q}(\boldsymbol{\theta})
=
0
\right\}.    
\end{equation}
The corresponding Jacobian matrix is given by
\[
D\boldsymbol{F}_{\boldsymbol c}(\boldsymbol{\theta})
=
D\boldsymbol{P}(\boldsymbol{\theta})
-
\boldsymbol{c}\,
D\boldsymbol{Q}(\boldsymbol{\theta}),
\]
which differs from the Jacobian of $\boldsymbol{\xi}$ in \eqref{eq:dxi} by the factor $\boldsymbol{Q}(\boldsymbol{\theta})^{-1}$. For generic $\boldsymbol{c}\in\mathcal{C}$, the level set $\mathcal{M}_{\boldsymbol{c}}$ forms a smooth algebraic manifold outside a singular algebraic subset of measure zero.
Sampling from algebraic manifolds has been studied through Crofton-type constructions, where samples are generated via intersections between the manifold and random affine subspaces of complementary dimension \citep{J33}. While such methods provide globally supported i.i.d. samples on $\mathcal{M}_{\boldsymbol{c}}$, the induced density $q_{\boldsymbol{c}}$ is generally known only up to an intractable normalizing constant depending on $\boldsymbol{c}$. Consequently, these constructions cannot be directly employed within the pseudo-marginal estimator in \eqref{eq:est}, which requires explicit evaluation of the proposal density.

To address this issue, we consider a coordinate partition approach for constructing a tractable proposal density $q_{\boldsymbol c}$ for the estimator $\widehat{\tilde{\pi}}_{\mathcal C}(\boldsymbol c)$. Since $\boldsymbol{\xi}$ is a submersion almost everywhere, the Jacobian matrix
\(
D\boldsymbol F_{\boldsymbol c}(\boldsymbol{\theta})
\)
has rank $q$ for almost every $\boldsymbol{\theta}\in\mathcal M_{\boldsymbol c}$. At each regular point, there exists a collection of $q$ coordinates whose associated Jacobian minor is nonsingular. We can therefore partition the parameter vector as
\(
\boldsymbol{\theta}
=
(
\boldsymbol{\theta}_{\mathcal I},
\boldsymbol{\theta}_{\mathcal D}
),
\)
where
\(
\boldsymbol{\theta}_{\mathcal I}\in\mathbb R^{p+n-q}\) and
\(\boldsymbol{\theta}_{\mathcal D}\in\mathbb R^{q},
\)
such that
\(
\frac{
\partial
\boldsymbol F_{\boldsymbol c}
}{
\partial
\boldsymbol\theta_{\mathcal D}
}
\)
is nonsingular. This induces the following local graph representation of the manifold $\mathcal M_{\boldsymbol c}$ in terms of the independent variables $\boldsymbol{\theta}_{\mathcal I}$, which forms the basis for constructing globally supported samplers through polynomial root finding algorithms.

\begin{lemma}\label{lem:3.2}
Let
\(
\boldsymbol{\theta}^{\ast}
=
(
\boldsymbol{\theta}_{\mathcal I}^{\ast},
\boldsymbol{\theta}_{\mathcal D}^{\ast}
)
\in
\mathcal M_{\boldsymbol c}
\)
be a regular point satisfying
\(
\det\left(
\frac{
\partial
\boldsymbol F_{\boldsymbol c}
}{
\partial
\boldsymbol\theta_{\mathcal D}
}
(\boldsymbol{\theta}^{\ast})
\right)
\neq 0.
\)
Then there exists an open neighborhood
\(
U\subset\mathbb R^{p+n-q}
\)
of
\(
\boldsymbol{\theta}_{\mathcal I}^{\ast}
\)
and a unique smooth mapping
\(
\boldsymbol g_{\boldsymbol c}:U\to\mathbb R^{q}
\)
s.t. 
\[
\boldsymbol g_{\boldsymbol c}
(
\boldsymbol{\theta}_{\mathcal I}^{\ast}
)
=
\boldsymbol{\theta}_{\mathcal D}^{\ast}, \quad
\boldsymbol F_{\boldsymbol c}
\big(
\boldsymbol{\theta}_{\mathcal I},
\boldsymbol g_{\boldsymbol c}
(
\boldsymbol{\theta}_{\mathcal I}
)
\big)
=
0
\]
for all
\(
\boldsymbol{\theta}_{\mathcal I}\in U
\).
Consequently, the manifold $\mathcal M_{\boldsymbol c}$ admits the local graph representation
\[
\mathcal M_{\boldsymbol c}\cap(U\times\mathbb R^q)
=
\left\{
(
\boldsymbol{\theta}_{\mathcal I},
\boldsymbol g_{\boldsymbol c}
(
\boldsymbol{\theta}_{\mathcal I}
)
):
\boldsymbol{\theta}_{\mathcal I}\in U
\right\}.
\]
\end{lemma}
\begin{proof}
Since 
\(
\boldsymbol{\theta}^{\ast}
=
(
\boldsymbol{\theta}_{\mathcal I}^{\ast},
\boldsymbol{\theta}_{\mathcal D}^{\ast}
)
\in
\mathcal M_{\boldsymbol c}
\),
we have
\(
\boldsymbol F_{\boldsymbol c}
(
\boldsymbol{\theta}_{\mathcal I}^{\ast},
\boldsymbol{\theta}_{\mathcal D}^{\ast}
)
=
0.
\)
By assumption
\(
\det\!\left(
\frac{\partial \boldsymbol F_{\boldsymbol c}}
{\partial \boldsymbol\theta_{\mathcal D}}
(
\boldsymbol{\theta}^{\ast}
)
\right)
\neq 0,
\) the Jacobian matrix of
\(
\boldsymbol F_{\boldsymbol c}
\)
with respect to the dependent variables
\(
\boldsymbol\theta_{\mathcal D}
\)
is invertible at
\(
\boldsymbol{\theta}^{\ast}
\).
Since
\(
\boldsymbol F_{\boldsymbol c}
\)
is continuously differentiable, the implicit function theorem implies that there exist open neighborhoods
\[
U \subset \mathbb R^{p+n-q},
\qquad
V \subset \mathbb R^{q},
\]
containing
\(
\boldsymbol{\theta}_{\mathcal I}^{\ast}
\)
and
\(
\boldsymbol{\theta}_{\mathcal D}^{\ast}
\),
respectively, together with a unique smooth mapping
\(
\boldsymbol g_{\boldsymbol c}:U\to V
\)
such that
\[
\boldsymbol g_{\boldsymbol c}
(
\boldsymbol{\theta}_{\mathcal I}^{\ast}
)
=
\boldsymbol{\theta}_{\mathcal D}^{\ast},
\quad
\boldsymbol F_{\boldsymbol c}
\big(
\boldsymbol{\theta}_{\mathcal I},
\boldsymbol g_{\boldsymbol c}
(
\boldsymbol{\theta}_{\mathcal I}
)
\big)
=
0
\]
for every
\(
\boldsymbol{\theta}_{\mathcal I}\in U
\). If
\(
(
\boldsymbol{\theta}_{\mathcal I},
\boldsymbol{\theta}_{\mathcal D}
)
\in
(U\times V)\cap\mathcal M_{\boldsymbol c},
\)
then
\(
\boldsymbol F_{\boldsymbol c}
(
\boldsymbol{\theta}_{\mathcal I},
\boldsymbol{\theta}_{\mathcal D}
)
=
0.
\)
By the uniqueness statement in the implicit function theorem,
\(
\boldsymbol{\theta}_{\mathcal D}
=
\boldsymbol g_{\boldsymbol c}
(
\boldsymbol{\theta}_{\mathcal I}
).
\)
Hence,
\[
\mathcal M_{\boldsymbol c}\cap(U\times V)
=
\left\{
(
\boldsymbol{\theta}_{\mathcal I},
\boldsymbol g_{\boldsymbol c}
(
\boldsymbol{\theta}_{\mathcal I}
)
):
\boldsymbol{\theta}_{\mathcal I}\in U
\right\},
\]
which proves the claimed local graph representation.
\end{proof}

The partition
\(
\boldsymbol{\theta}
=
(
\boldsymbol{\theta}_{\mathcal I},
\boldsymbol{\theta}_{\mathcal D}
)
\)
between independent and dependent variables can be chosen arbitrarily, provided that the Jacobian submatrix
\(
\frac{
\partial
\boldsymbol F_{\boldsymbol c}
}{
\partial
\boldsymbol\theta_{\mathcal D}
}
(\boldsymbol{\theta})
\)
is nonsingular at the point under consideration. By Lemma~\ref{lem:3.2}, this condition guarantees that the manifold
\(
\mathcal M_{\boldsymbol c}
\)
admits a local graph representation over the coordinates
\(
\boldsymbol{\theta}_{\mathcal I}
\).
For a fixed coordinate partition, define the singular set
\[
\mathcal S
:=
\left\{
\boldsymbol{\theta}\in\mathcal M_{\boldsymbol c}
:
\det\!\left(
\frac{
\partial
\boldsymbol F_{\boldsymbol c}
}{
\partial
\boldsymbol\theta_{\mathcal D}
}
(\boldsymbol{\theta})
\right)
=0
\right\}.
\]
Since $\boldsymbol F_{\boldsymbol c}$ is polynomial, the determinant above is a smooth algebraic function on
\(
\mathcal M_{\boldsymbol c}
\).
If this determinant is not identically zero, then
\(
\mathcal S
\)
forms a proper algebraic subset of
\(
\mathcal M_{\boldsymbol c}
\). It has lower dimension and Hausdorff measure zero on
\(
\mathcal M_{\boldsymbol c}
\).
Hence, the corresponding coordinate chart is valid almost everywhere on the manifold, and the singular set
\(
\mathcal S
\)
does not affect validity of the sampling procedure. Although any partitions satisfying the nonsingularity condition is theoretically sufficient, the numerical stability and efficiency of the  subsequent sampling procedure depend on the choice of coordinates. In practice, it is advantageous to select the independent variables so that the corresponding Jacobian minor remains as well-conditioned as possible. This improves the stability of the local graph representation, reduces geometric distortion in the induced proposal density on the manifold, and enhances the robustness of the polynomial root finding algorithm.

We first consider a global partition over the parameter space $\Theta$ based on the prior sensitivity matrix
\[
\boldsymbol S
=
\int
D\boldsymbol F_{\boldsymbol c}(\boldsymbol{\theta})^\top
D\boldsymbol F_{\boldsymbol c}(\boldsymbol{\theta})
\,p(\boldsymbol{\theta})
\,d\boldsymbol{\theta},
\]
where $p$ is the prior density given in \eqref{eq:prior}. The matrix
\(
\boldsymbol S
\)
captures the average sensitivity and linear dependence structure of the parameters under the prior distribution.
A global partition can then be obtained using a rank-revealing QR (RRQR) factorization
\[
\boldsymbol S\boldsymbol P
=
\boldsymbol Q
\begin{bmatrix}
\boldsymbol R_{11} & \boldsymbol R_{12}
\end{bmatrix},
\]
where the permutation matrix $\boldsymbol{P}$ reorders the coordinates according to their numerical linear independence \citep{J34}. The submatrix $\boldsymbol{R}_{11}$ is a upper triangular block corresponding to the $q$ most linearly independent columns, while $\boldsymbol{R}_{12}$ accounts for the remaining dependencies. The first
\(
q
\)
pivot coordinates are assigned as dependent variables
\(
\boldsymbol{\theta}_{\mathcal D}
\),
while the remaining
\(
p+n-q
\)
coordinates define the independent variables
\(
\boldsymbol{\theta}_{\mathcal I}
\).
This produces a globally well-conditioned coordinate partition except on a measure zero singular set
\(
\mathcal S
\). In practice, the sensitivity matrix $\boldsymbol S$ can be approximated using Monte Carlo integration with i.i.d. samples drawn from the prior distribution. A RRQR decomposition can then be applied to the resulting empirical matrix in order to find the coordinate partition. For a detailed discussion of this approximation method, we refer to \citet{J35}.

An alternative approach is to update the RRQR partition for each proposed value of $\boldsymbol c$, and hence for each algebraic manifold $\mathcal M_{\boldsymbol c}$. This produces a locally optimized coordinate partition tailored to the geometry of the corresponding manifold, potentially improving numerical conditioning and reducing geometric distortion in the induced proposal density. Such a strategy is advantageous near singular regions or highly curved components of the manifold, where a fixed global partition may become nearly degenerate. However, repartitioning introduces additional computational cost, since a new Jacobian factorization must be computed at each iteration. In what follows, we assume by default the use of a fixed partition obtained from the prior sensitivity matrix, as it provides a simpler and computationally efficient implementation. Nevertheless, repartitioning can be incorporated directly when improved local conditioning is required in practical implementations.

We now discuss how to construct an analytical proposal density
\(
q_{\boldsymbol c}
\)
for the pseudo-marginal estimator in \eqref{eq:est}. Let
\(
\boldsymbol{\theta}
=
(
\boldsymbol{\theta}_{\mathcal I},
\boldsymbol{\theta}_{\mathcal D}
)
\)
be a fixed global coordinate partition of the parameter space, where
\(
\boldsymbol{\theta}_{\mathcal I}\in\mathbb R^{p+n-q}
\)
denotes the independent coordinates and
\(
\boldsymbol{\theta}_{\mathcal D}\in\mathbb R^q
\)
denotes the dependent coordinates. Let
\(
g_{\boldsymbol c}
(
\boldsymbol{\theta}_{\mathcal I}
)
\)
be a tractable probability density fully supported on the projection of
\(
\Theta
\)
onto the independent coordinate space.
The coordinate partition sampling procedure generates manifold samples in two stages. First, coordinates
\(
\boldsymbol{\theta}_{\mathcal I}^{(i)}
\sim
g_{\boldsymbol c}\), \(
i=1,\dots,N,
\)
are sampled independently. Conditional on
\(
\boldsymbol{\theta}_{\mathcal I}^{(i)}
\),
the polynomial system
\(
\boldsymbol F_{\boldsymbol c}
(
\boldsymbol{\theta}_{\mathcal I}^{(i)},
\boldsymbol{\theta}_{\mathcal D}
)
=
\boldsymbol 0
\)
is solved using a global polynomial root solver to recover all real solutions
\[
\mathcal S
(
\boldsymbol{\theta}_{\mathcal I}^{(i)},
\boldsymbol c
)
=
\left\{
\boldsymbol{\theta}_{\mathcal D}\in\mathbb R^q:
\boldsymbol F_{\boldsymbol c}
(
\boldsymbol{\theta}_{\mathcal I}^{(i)},
\boldsymbol{\theta}_{\mathcal D}
)
=
\boldsymbol 0,
\;
(
\boldsymbol{\theta}_{\mathcal I}^{(i)},
\boldsymbol{\theta}_{\mathcal D}
)
\in\Theta
\right\}.
\]
If the root set $\mathcal{S}$ is nonempty, one select uniformly from
it to obtain a manifold sample
\(
\boldsymbol{\theta}^{(i)}
=
(
\boldsymbol{\theta}_{\mathcal I}^{(i)},
\boldsymbol{\theta}_{\mathcal D}^{(i)}
)
\in
\mathcal M_{\boldsymbol c}.
\)
The induced proposal density
\(
q_{\boldsymbol c}
\)
is characterized by the following lemma.
\begin{lemma}\label{lem:coord_partition}
Let $\mathcal M_{\boldsymbol c}$ be the manifold defined by
\(
\boldsymbol F_{\boldsymbol c}(\boldsymbol{\theta})=0,
\)
and let
\(
K(\boldsymbol{\theta}_{\mathcal I})
=
\left|
\mathcal S
(
\boldsymbol{\theta}_{\mathcal I},
\boldsymbol c
)
\right|
\)
denote the number of admissible real roots associated with
\(
\boldsymbol{\theta}_{\mathcal I}
\).
Then the coordinate partition sampling procedure induces a probability density
\(
q_{\boldsymbol c}
\)
with respect to the Hausdorff measure
\(
\mathcal H^{p+n-q}
\)
on $\mathcal M_{\boldsymbol c}$ given by
\[
q_{\boldsymbol c}(\boldsymbol{\theta})
=
\frac{
g_{\boldsymbol{c}}(\boldsymbol{\theta}_{\mathcal I})
}{
K(\boldsymbol{\theta}_{\mathcal I})
}
\,
\frac{
\left|
\det
\left(
\frac{
\partial
\boldsymbol F_{\boldsymbol c}
}{
\partial
\boldsymbol\theta_{\mathcal D}
}
(\boldsymbol{\theta})
\right)
\right|
}{
J_{\boldsymbol F_{\boldsymbol c}}(\boldsymbol{\theta})
},
\]
where
\[
J_{\boldsymbol F_{\boldsymbol c}}(\boldsymbol{\theta})
=
\sqrt{
\det
\left(
D\boldsymbol F_{\boldsymbol c}(\boldsymbol{\theta})
D\boldsymbol F_{\boldsymbol c}(\boldsymbol{\theta})^\top
\right)
}.
\]
Furthermore, if
\(
g_{\boldsymbol{c}}(\boldsymbol{\theta}_{\mathcal I})>0\)
for all 
\(\boldsymbol{\theta}_{\mathcal I}\in\Theta_{\mathcal I},
\)
then
\(
q_{\boldsymbol c}(\boldsymbol{\theta})>0
\)
for almost every
\(
\boldsymbol{\theta}\in\mathcal M_{\boldsymbol c}
\),
and the sampler has full support on $\mathcal M_{\boldsymbol c}$, including disconnected components.
\end{lemma}

\begin{proof}
By the coarea formula, the Hausdorff measure on
\(
\mathcal M_{\boldsymbol c}
\)
satisfies
\[
d\mathcal H^{p+n-q}(\boldsymbol{\theta})
=
\frac{
J_{\boldsymbol F_{\boldsymbol c}}(\boldsymbol{\theta})
}{
\left|
\det
\left(
\frac{
\partial
\boldsymbol F_{\boldsymbol c}
}{
\partial
\boldsymbol\theta_{\mathcal D}
}
(\boldsymbol{\theta})
\right)
\right|
}
\,d\boldsymbol{\theta}_{\mathcal I}.
\]
Under the sampling construction, the probability of selecting a point
\[
\boldsymbol{\theta}
=
(
\boldsymbol{\theta}_{\mathcal I},
\boldsymbol{\theta}_{\mathcal D}
)
\in
\mathcal M_{\boldsymbol c}
\]
is obtained by first sampling
\(
\boldsymbol{\theta}_{\mathcal I}
\)
from
\(
g_{\boldsymbol{c}}
\),
followed by uniform selection among the
\(
K(\boldsymbol{\theta}_{\mathcal I})
\)
admissible roots. Consequently,
\[
dQ_{\boldsymbol c}(\boldsymbol{\theta})
=
\frac{
g_{\boldsymbol{c}}(\boldsymbol{\theta}_{\mathcal I})
}{
K(\boldsymbol{\theta}_{\mathcal I})
}
\,d\boldsymbol{\theta}_{\mathcal I}.
\]
Substituting the coarea relation yields
\[
dQ_{\boldsymbol c}(\boldsymbol{\theta})
=
\frac{
g_{\boldsymbol{c}}(\boldsymbol{\theta}_{\mathcal I})
}{
K(\boldsymbol{\theta}_{\mathcal I})
}
\,
\frac{
\left|
\det
\left(
\frac{
\partial
\boldsymbol F_{\boldsymbol c}
}{
\partial
\boldsymbol\theta_{\mathcal D}
}
(\boldsymbol{\theta})
\right)
\right|
}{
J_{\boldsymbol F_{\boldsymbol c}}(\boldsymbol{\theta})
}
\,d\mathcal H^{p+n-q}(\boldsymbol{\theta}),
\]
which proves the density formula.
Finally, since
\(
g_{\boldsymbol{c}}(\boldsymbol{\theta}_{\mathcal I})>0
\)
on
\(
\Theta_{\mathcal I}
\),
every admissible projection has positive sampling probability. The global polynomial solver recovers all admissible real roots of the polynomial system, implying that every connected component of
\(
\mathcal M_{\boldsymbol c}
\)
is reachable with positive probability. Therefore,
\(
q_{\boldsymbol c}
\)
has full support on
\(
\mathcal M_{\boldsymbol c}
\).
\end{proof}
Substituting the explicit density
\(
q_{\boldsymbol c}
\)
into the pseudo-marginal estimator in \eqref{eq:est} leads to cancellation of the Jacobian terms
\(
J_{\boldsymbol F_{\boldsymbol c}}
\)
and
\(
J_{\boldsymbol \xi}
\),
resulting in a simplified algebraic form for the importance weights. The coordinate partition sampling procedure gives a tractable pseudo-marginal estimator for the unnormalized posterior
\(
\tilde{\pi}_{\mathcal C}
\).
The complete process is summarized in Algorithm~\ref{alg:coord_partition}.

\begin{algorithm}[H]
\caption{Coordinate Partition Estimator for $\widehat{\tilde{\pi}}_{\mathcal C}(\boldsymbol{c}, U)$}
\label{alg:coord_partition}
\begin{algorithmic}[1]
\Require Identifiable parameter combination $\boldsymbol c=\boldsymbol{\xi}(\boldsymbol{\theta})$, coordinate partition
\(
\boldsymbol{\theta}
=
(
\boldsymbol{\theta}_{\mathcal I},
\boldsymbol{\theta}_{\mathcal D}
)
\),
proposal density
\(
g_{\boldsymbol c}
\),
sample size \(N\)

\For{$i=1,\dots,N$}

\State Sample
\(
\boldsymbol{\theta}_{\mathcal I}^{(i)}
\sim
g_{\boldsymbol c}.
\)
\State Solve $\boldsymbol F_{\boldsymbol c}(\boldsymbol{\theta}_{\mathcal I}^{(i)}, \boldsymbol{\theta}_{\mathcal D}) = \boldsymbol 0$ with a global polynomial solver (e.g., \citealt{J32}).
\State Construct
\(
\mathcal S
(
\boldsymbol{\theta}_{\mathcal I}^{(i)},
\boldsymbol c
)
=
\left\{
\boldsymbol{\theta}_{\mathcal D}\in\mathbb R^q:
\boldsymbol F_{\boldsymbol c}
(
\boldsymbol{\theta}_{\mathcal I}^{(i)},
\boldsymbol{\theta}_{\mathcal D}
)
=
\boldsymbol 0,
\;
(
\boldsymbol{\theta}_{\mathcal I}^{(i)},
\boldsymbol{\theta}_{\mathcal D}
)
\in\Theta
\right\}.
\)

\State Compute
\(
K_i
=
\left|
\mathcal S
(
\boldsymbol{\theta}_{\mathcal I}^{(i)},
\boldsymbol c
)
\right|.
\)

\If{$K_i>0$}

\State Sample
\(
\boldsymbol{\theta}_{\mathcal D}^{(i)}
\sim
\mathrm{Unif}
\left(
\mathcal S
(
\boldsymbol{\theta}_{\mathcal I}^{(i)},
\boldsymbol c
)
\right).
\)

\State Set
\(
\boldsymbol{\theta}^{(i)}
=
(
\boldsymbol{\theta}_{\mathcal I}^{(i)},
\boldsymbol{\theta}_{\mathcal D}^{(i)}
).
\)

\State Compute
\(
w_i
=
\tilde{\pi}(\boldsymbol{\theta}^{(i)})
/
J_{\boldsymbol{\xi}}(\boldsymbol{\theta}^{(i)})q_{\boldsymbol c}(\boldsymbol{\theta}^{(i)})
.
\)

\Else

\State Set
\(
w_i=0.
\)

\EndIf

\EndFor
\State \Return Pseudo-marginal estimator
\(
\widehat{\tilde{\pi}}_{\mathcal C}(\boldsymbol c,U)
=
\frac1N
\sum_{i=1}^N
w_i.
\)
\end{algorithmic}
\end{algorithm}
Though the choice of
\(
g_{\boldsymbol c}
\)
is theoretically arbitrary provided it has full support on
the projection of
\(
\Theta
\)
onto the independent coordinate,
it influences the efficiency of the resulting sampling procedure.
In many applications, the parameter space
\(
\Theta
\subset
\mathbb R^{p+n}
\)
is closed, compact, and connected. Each coordinate
\(
\theta_i
\)
admits finite lower and upper bounds. Since the coordinate partition sampler requires solving a polynomial system for each sampled value of
\(
\boldsymbol{\theta}_{\mathcal I}
\),
the proposal density
\(
g_{\boldsymbol c}
\)
should ideally maximize the probability that the resulting fibre intersects the manifold
\(
\mathcal M_{\boldsymbol c}
\)
within the parameter space
\(
\Theta
\).
A simple choice is to sample uniformly over the marginal bounds of the independent coordinates
\(
\boldsymbol{\theta}_{\mathcal I}
\).

\paragraph{Example 1: Uniform proposal over marginal bounds.}
Suppose the ranges of the independent coordinates are given by
\[
\theta_{\mathcal I,j}\in[a_j,b_j],
\qquad
j=1,\dots,p+n-q.
\]
Define
\(
\Theta_{\mathcal I}
=
\prod_{j=1}^{p+n-q}[a_j,b_j],
\)
a natural baseline proposal is then the uniform density
\[
g_{\boldsymbol c}
(
\boldsymbol{\theta}_{\mathcal I}
)
=
\frac{
1
}{
|\Theta_{\mathcal I}|
},
\qquad
\boldsymbol{\theta}_{\mathcal I}\in\Theta_{\mathcal I},
\]
where
\(
|\Theta_{\mathcal I}|
\)
denotes the Lebesgue measure of
\(
\Theta_{\mathcal I}
\).
By Lemma~\ref{lem:coord_partition}, the induced manifold density becomes
\[
q_{\boldsymbol c}(\boldsymbol{\theta})
=
\frac{
1
}{
|\Theta_{\mathcal I}|
\,K(\boldsymbol{\theta}_{\mathcal I})
}
\,
\frac{
\left|
\det
\left(
\frac{
\partial
\boldsymbol F_{\boldsymbol c}
}{
\partial
\boldsymbol\theta_{\mathcal D}
}
(\boldsymbol{\theta})
\right)
\right|
}{
J_{\boldsymbol F_{\boldsymbol c}}(\boldsymbol{\theta})
}.
\]
The advantage of this construction is its simplicity and global coverage. Since proposals are generated uniformly over the coordinate bounds, the method avoids a great portion of samples lying outside the $\Theta$ and allows exploration across all connected components of
\(
\mathcal M_{\boldsymbol c}
\).
The proposal does not account for whether a sampled projection
\(
\boldsymbol{\theta}_{\mathcal I}
\)
gives real roots of the polynomial system. Consequently, some sampled projections may correspond to fibres that do not intersect
\(
\mathcal M_{\boldsymbol c}
\)
within
\(
\Theta
\),
especially in high-dimensional settings or when the feasible manifold occupies only a small subset of the ambient parameter space. 

\paragraph{Example 2: Local Gaussian proposal.}
To improve computational efficiency, we construct a local proposal density centered near the manifold
\(
\mathcal M_{\boldsymbol c}
\).
Suppose
\(
\boldsymbol{\theta}'
=
(
\boldsymbol{\theta}'_{\mathcal I},
\boldsymbol{\theta}'_{\mathcal D}
)
\in
\mathcal M_{\boldsymbol c'}
\)
is a sample obtained from a previous pseudo-marginal MCMC iteration corresponding to
\(
\boldsymbol c'
\).
Since the polynomial mapping
\(
\boldsymbol F_{\boldsymbol c}
\)
depends smoothly on
\(
\boldsymbol c
\),
the manifolds
\(
\mathcal M_{\boldsymbol c}
\)
and
\(
\mathcal M_{\boldsymbol c'}
\)
vary continuously for nearby identifiable parameter combinations. Consequently,
\(
\boldsymbol{\theta}'
\)
provides a reference point for constructing proposals on
\(
\mathcal M_{\boldsymbol c}
\).
We generate local proposals by perturbing the independent coordinates
\(
\boldsymbol{\theta}_{\mathcal I}
\)
near
\(
\boldsymbol{\theta}'_{\mathcal I}
\).
To account for the local geometry, we adapt the proposal covariance using the Jacobian of the polynomial system with respect to the independent coordinates,
\[
D_{\mathcal I}\boldsymbol F_{\boldsymbol c}
(
\boldsymbol{\theta}'
)
:=
\frac{
\partial
\boldsymbol F_{\boldsymbol c}
}{
\partial
\boldsymbol\theta_{\mathcal I}
}
(
\boldsymbol{\theta}'
)
\in
\mathbb R^{q\times(p+n-q)}.
\]
The Gram matrix
\[
G(
\boldsymbol{\theta}'
)
=
D_{\mathcal I}\boldsymbol F_{\boldsymbol c}
(
\boldsymbol{\theta}'
)
^\top
D_{\mathcal I}\boldsymbol F_{\boldsymbol c}
(
\boldsymbol{\theta}'
)
\]
defines a local metric induced by the constraint map. Directions corresponding to large eigenvalues of
\(
G
\)
produce large first order variations in
\(
\boldsymbol F_{\boldsymbol c}
\). We define the proposal covariance matrix by
\[
\underline{\Sigma}
:=
\epsilon_1
G^{-1}(
\boldsymbol{\theta}'
)
+
\epsilon_2 I,
\]
where
\(
\epsilon_1>0
\)
controls the extent of exploration in local geometry, and
\(
\epsilon_2>0
\)
adds isotropic regularization. This construction increases the probability of obtaining real roots for the polynomial solver. Smaller values of
\(
\epsilon_1
\)
favor local exploration and higher root-finding efficiency, while larger values give broader exploration across the parameter space.
The resulting Gaussian proposal density on the independent coordinates is
\[
g_{\boldsymbol c}
(
\boldsymbol{\theta}_{\mathcal I}
)
=
\frac{
1
}{
(2\pi)^{(p+n-q)/2}
|\underline{\Sigma}|^{1/2}
}
\exp\!\left(
-
\frac12
(
\boldsymbol{\theta}_{\mathcal I}
-
\boldsymbol{\theta}'_{\mathcal I}
)^\top
\underline{\Sigma}^{-1}
(
\boldsymbol{\theta}_{\mathcal I}
-
\boldsymbol{\theta}'_{\mathcal I}
)
\right).
\]
The induced proposal density on the manifold is then given by
\[
q_{\boldsymbol c}(\boldsymbol{\theta})
=
\frac{
g_{\boldsymbol c}
(
\boldsymbol{\theta}_{\mathcal I}
)
}{
K(\boldsymbol{\theta}_{\mathcal I})
}
\,
\frac{
\left|
\det
\left(
\frac{
\partial
\boldsymbol F_{\boldsymbol c}
}{
\partial
\boldsymbol\theta_{\mathcal D}
}
(\boldsymbol{\theta})
\right)
\right|
}{
J_{\boldsymbol F_{\boldsymbol c}}(\boldsymbol{\theta})
}.
\]
Compared with a global uniform proposal, this local Gaussian construction concentrates probability mass near regions where real roots are more likely to occur. As a result, the number of failed polynomial solves is reduced, leading to improved computational efficiency.

From these examples, the choice of
\(
g_{\boldsymbol c}
\)
induces a trade-off between global geometric exploration and computational efficiency. Uniform proposals provide global coverage of the algebraic manifold, whereas localized Gaussian proposals improve efficiency by increasing the probability of obtaining real roots. In practice, uniform proposals are sufficient for low-dimensional or geometrically simple manifolds, while local proposals become advantageous for high-dimensional or geometrically complex algebraic manifolds.
The final step for a pseudo-marginal sampler is to reconstruct a representative parameter sample from i.i.d. conditional samples drawn from
\(
q_{\boldsymbol c}
\)
on the non-identifiable algebraic manifold
\(
\mathcal M_{\boldsymbol c}
\).

\subsubsection{Reconstructing the Parameter Sample}

In the pseudo-marginal framework, the state of the Markov chain is given by the augmented variable $(\boldsymbol{c},U)$, where
\(
U=\{\boldsymbol{\theta}^{(1)},\dots,\boldsymbol{\theta}^{(N)}\}
\)
denotes a collection of i.i.d. samples drawn from the proposal distribution $q_{\boldsymbol c}$ on the manifold $\mathcal M_{\boldsymbol c}$. After a proposed level set $\boldsymbol c$ is accepted in the Metropolis--Hastings step, a representative parameter sample $\boldsymbol{\theta}^*$ can be reconstructed from the auxiliary sample set $U$ in order to obtain a sample from the target posterior distribution. We perform this reconstruction using an importance resampling procedure.
For each sample $\boldsymbol{\theta}^{(i)}\in U$, we define the unnormalized importance weight
\[
w(\boldsymbol{\theta}^{(i)})
=
\frac{\tilde{\pi}(\boldsymbol{\theta}^{(i)})}
     {q_{\boldsymbol c}(\boldsymbol{\theta}^{(i)})} = \frac{L(\boldsymbol{\theta^{(i)}};D)p(\boldsymbol{\theta}^{(i)})}{q_{\boldsymbol{c}}(\boldsymbol{\theta}^{(i)})},
\qquad i=1,\dots,N.
\]
The weights are then normalized according to
\(
\bar w(\boldsymbol{\theta}^{(i)})
=
w(\boldsymbol{\theta}^{(i)})/\sum_{j=1}^N w(\boldsymbol{\theta}^{(j)}),
\)
which defines a probability distribution on the sample set $U$.
The representative parameter sample
\(
\boldsymbol{\theta}^*
\)
is then obtained by categorical resampling
\begin{equation}\label{eq:recon}
\boldsymbol{\theta}^*
\sim
\sum_{i=1}^N
\bar w(\boldsymbol{\theta}^{(i)})
\,
\delta_{\boldsymbol{\theta}^{(i)}},
\end{equation}
where
\(
\delta_{\boldsymbol{\theta}^{(i)}}
\)
denotes the Dirac probability measure concentrated at
\(
\boldsymbol{\theta}^{(i)}
\).
This procedure selects points with high posterior density relative to the proposal distribution. The following Lemma \ref{lem:3.4} shows that the resulting reconstruction step asymptotically recovers the conditional posterior distribution on the manifold.
\begin{lemma}\label{lem:3.4}
Suppose the importance weight function $w(\boldsymbol{\theta}) = \tilde{\pi}(\boldsymbol{\theta})/q_{\boldsymbol{c}}(\boldsymbol{\theta})$
satisfies
\begin{equation}\label{eq:L2-condition}
  \int_{\mathcal{M}_{\boldsymbol{c}}}
    \frac{\tilde{\pi}(\boldsymbol{\theta})^{2}}{q_{\boldsymbol{c}}(\boldsymbol{\theta})}
  \,\mathrm{d}\boldsymbol{\theta}
  \;<\; \infty.
\end{equation}
Let $U = \{\boldsymbol{\theta}^{(1)},\dots,\boldsymbol{\theta}^{(N)}\}$ be i.i.d.\ samples from
$q_{\boldsymbol{c}}$ on $\mathcal{M}_{\boldsymbol{c}}$, and let $\boldsymbol{\theta}^{*}$ be
drawn according to \eqref{eq:recon}.
Denote by $\widehat{\pi}_{N}(\cdot\mid\boldsymbol{c})$ the resulting distribution of
$\boldsymbol{\theta}^{*}$ given $\boldsymbol{c}$.  Then there exists a finite constant
$C(\boldsymbol{c})$, depending only on the second moment of $w$ under
$q_{\boldsymbol{c}}$, such that
\[
  d_{\mathrm{TV}}\!\bigl(\widehat{\pi}_{N}(\cdot\mid\boldsymbol{c}),\;
                          \pi(\cdot\mid\boldsymbol{c})\bigr)
  \;\leq\;
  \frac{C(\boldsymbol{c})}{\sqrt{N}}.
\]
\end{lemma}

\begin{proof}

By the dual representation of total variation distance,
\[
d_{\mathrm{TV}}
\bigl(
\widehat{\pi}_{N}(\cdot\mid\boldsymbol c),
\pi(\cdot\mid\boldsymbol c)
\bigr)
=
\sup_{\|\varphi\|_\infty\le1}
\Biggl|
\mathbb E
\bigl[
\varphi(\boldsymbol\theta^*)
\mid
\boldsymbol c
\bigr]
-
\int_{\mathcal M_{\boldsymbol c}}
\varphi(\boldsymbol\theta)
\,
\pi(\boldsymbol\theta\mid\boldsymbol c)
\,d\boldsymbol\theta
\Biggr|.
\]
It suffices to bound the right-hand side uniformly over all bounded test functions
\(
\varphi
\)
with
\(
\|\varphi\|_\infty\le1
\).
Define
\[
A_N
=
\frac1N
\sum_{i=1}^N
w(\boldsymbol\theta^{(i)})
\,
\varphi(\boldsymbol\theta^{(i)}),
\qquad
B_N
=
\frac1N
\sum_{i=1}^N
w(\boldsymbol\theta^{(i)}),
\]
and
\[
A
=
\int_{\mathcal M_{\boldsymbol c}}
\tilde\pi(\boldsymbol\theta)
\,
\varphi(\boldsymbol\theta)
\,d\boldsymbol\theta,
\qquad
B
=
\int_{\mathcal M_{\boldsymbol c}}
\tilde\pi(\boldsymbol\theta)
\,d\boldsymbol\theta
=
Z_{\boldsymbol c}>0.
\]
Condition~\eqref{eq:L2-condition} implies that
\(
w\in L^2(q_{\boldsymbol c})
\).
Since
\(
\|\varphi\|_\infty\le1
\),
we also have
\(
\operatorname{Var}_{q_{\boldsymbol c}}
\bigl[
w\varphi
\bigr]
\le
\operatorname{Var}_{q_{\boldsymbol c}}
[w]
<
\infty.
\)
Standard \(L^2\) bounds for empirical averages give
\[
\mathbb E
|A_N-A|
\le
\frac{
\sqrt{
\operatorname{Var}_{q_{\boldsymbol c}}
[w\varphi]
}
}{\sqrt N},
\qquad
\mathbb E
|B_N-B|
\le
\frac{
\sqrt{
\operatorname{Var}_{q_{\boldsymbol c}}
[w]
}
}{\sqrt N}.
\]
Since
\(
\mathbb E
\bigl[
\varphi(\boldsymbol\theta^*)
\mid U
\bigr]
=
A_N/B_N,
\)
we write
\[
\frac{A_N}{B_N}
-
\frac{A}{B}
=
\frac{
(A_N-A)B
-
A(B_N-B)
}{
BB_N
}.
\]
By the strong law of large numbers,
\(
B_N\to B=Z_{\boldsymbol c}>0
\)
almost surely. Hence, for sufficiently large
\(
N
\),
we have
\(
B_N\ge Z_{\boldsymbol c}/2
\)
with probability
\(
1-O(N^{-1})
\).
Using
\(
|A|\le B
\),
which follows from
\(
\|\varphi\|_\infty\le1
\),
gives
\[
\left|
\frac{A_N}{B_N}
-
\frac{A}{B}
\right|
\le
\frac{2}{Z_{\boldsymbol c}}
\Bigl(
|A_N-A|
+
|B_N-B|
\Bigr).
\]
Taking expectations and absorbing the negligible complementary event into the constant gives
\[
\mathbb E
\left|
\frac{A_N}{B_N}
-
\frac{A}{B}
\right|
\le
\frac{
2
}{
Z_{\boldsymbol c}\sqrt N
}
\left(
\sqrt{
\operatorname{Var}_{q_{\boldsymbol c}}
[w\varphi]
}
+
\sqrt{
\operatorname{Var}_{q_{\boldsymbol c}}
[w]
}
\right)
+
O(N^{-1}).
\]
Since
\[
\frac{A}{B}
=
\int_{\mathcal M_{\boldsymbol c}}
\varphi(\boldsymbol\theta)
\,
\pi(\boldsymbol\theta\mid\boldsymbol c)
\,d\boldsymbol\theta,
\]
taking the supremum over all
\(
\|\varphi\|_\infty\le1
\)
completes the proof.
\end{proof}
We now present the full Algorithm~\ref{alg:IAPMCMC} for pseudo-marginal MCMC sampling in structurally non-identifiable models. Unlike standard pseudo-marginal methods, the estimator \(\widehat{\pi}_{\mathcal C}\) in Line 4 does not require repeated forward-model evaluations. Since structural non-identifiability implies that all parameter values on \(\mathcal M_{\boldsymbol c}\) share the same likelihood, only the prior density on the manifold must be evaluated through the corresponding importance weights. The convergence properties of Algorithm~\ref{alg:IAPMCMC} depend on the geometric structure of the identifiable mapping
\(
\boldsymbol{\xi}:\Theta\to\mathcal C
\).
Let
\(
\Theta\subseteq\mathbb R^{p+n}
\)
be compact and connected, and define the unnormalised marginal density
\[
\tilde{\pi}_{\mathcal C}(\boldsymbol c)
:=
\int_{\mathcal M_{\boldsymbol c}}
\tilde{\pi}(\boldsymbol\theta)
\,d\boldsymbol\theta.
\]
Let
\(
\widetilde P
\)
denote the pseudo-marginal Markov kernel on the augmented space
\(
\mathcal X:=\mathcal C\times\Theta^N
\).
For
\(
U
=
\{
\boldsymbol\theta^{(1)},
\dots,
\boldsymbol\theta^{(N)}
\},
\)
define
\[
\mu_{N,\boldsymbol c}(dU)
=
\prod_{i=1}^N
q_{\boldsymbol c}(\boldsymbol\theta^{(i)})
\,d\boldsymbol\theta^{(i)}.
\]
Then
\begin{align}
\nonumber
\widetilde P
\bigl(
(\boldsymbol c,U),
d(\boldsymbol c',U')
\bigr)
=
\widetilde\alpha
\bigl(
\boldsymbol c,U;
\boldsymbol c',U'
\bigr)
\,q(\boldsymbol c'\mid\boldsymbol c)
\,\mu_{N,\boldsymbol c'}(dU')
\,d\boldsymbol c'
\nonumber+
r(\boldsymbol c,U)
\,\delta_{(\boldsymbol c,U)}
\bigl(
d(\boldsymbol c',U')
\bigr),
\end{align}
where
\(
\widetilde\alpha
\)
is given by~\eqref{eq:alpha}, and
\[
r(\boldsymbol c,U)
:=
1-
\int_{\mathcal C}
\int_{\Theta^N}
\widetilde\alpha
\bigl(
\boldsymbol c,U;
\boldsymbol c',U'
\bigr)
\,q(\boldsymbol c'\mid\boldsymbol c)
\,\mu_{N,\boldsymbol c'}(dU')
\,d\boldsymbol c'.
\]
Under the following assumptions, the pseudo-marginal chain is Harris ergodic and converges in total variation. The result is stated in Theorem~\ref{thm:weak-tv}.

\begin{assumption}
\label{ass:rational}

We impose the following assumptions on the model
\(
\Sigma(\boldsymbol{\theta})
\)
and the inputs of Algorithm~\ref{alg:IAPMCMC}.

\begin{enumerate}

\item[(i)]
\(
\Theta\subseteq\mathbb R^d
\)
is compact and connected, and
\(
\boldsymbol\xi:\Theta\to\mathcal C
\)
is a rational map of generic rank \(p\).

\item[(ii)]
For
\(
\tilde\pi_{\mathcal C}
\)-a.e.
\(
\boldsymbol c\in\mathcal C
\),
the estimator
\(
\widehat{\tilde\pi}_{\mathcal C}
(\boldsymbol c,U)
\)
is unbiased with finite variance.

\item[(iii)]
The constant
\(
C(\boldsymbol c)
\)
in Lemma~\ref{lem:3.4} has finite expectation with respect to
\(
\pi_{\mathcal C}
\), i.e.
\[
\int_{\mathcal C}
C(\boldsymbol c)
\,\pi_{\mathcal C}(\boldsymbol c)
\,d\boldsymbol c
<
\infty.
\]

\item[(iv)]
\(
q(\boldsymbol c'\mid\boldsymbol c)
\)
is strictly positive and measurable on
\(
\mathcal C\times\mathcal C
\),
and the induced Metropolis--Hastings chain on
\(
\mathcal C
\)
is irreducible and aperiodic.

\end{enumerate}

\end{assumption}

\begin{theorem}
\label{thm:weak-tv}

Under Assumption~\ref{ass:rational},
let
\(
\{
(\boldsymbol c_k,U_k)
\}_{k\ge0}
\)
be the pseudo-marginal Markov chain generated by
Algorithm~\ref{alg:IAPMCMC},
and let
\(
\boldsymbol\theta_k
\)
denote the reconstructed parameter sample at iteration
\(
k
\).
Then
\begin{enumerate}

\item[(i)]
The augmented pseudo-marginal kernel
\(
\widetilde P
\)
admits the invariant probability measure
\[
\Pi(d\boldsymbol c,dU)
\propto
\tilde\pi_{\mathcal C}(\boldsymbol c)
\,
\mu_{N,\boldsymbol c}(dU).
\]

\item[(ii)]
The augmented chain is Harris ergodic, i.e. for every initial state
\(
(\boldsymbol c_0,U_0)\in\mathcal X
\),
\[
d_{\mathrm{TV}}
\Bigl(
\widetilde P^k
((\boldsymbol c_0,U_0),\cdot),
\Pi
\Bigr)
\xrightarrow[k\to\infty]{}
0.
\]

\item[(iii)]
For every
\(
k\ge0
\),
the probability distribution of
\(
\boldsymbol\theta_k
\)
satisfies
\begin{align}\nonumber
\label{eq:weak-tv}
d_{\mathrm{TV}}
\bigl(
\mathcal L(\boldsymbol\theta_k),
\pi
\bigr)
\le
d_{\mathrm{TV}}
\bigl(
\mathcal L(\boldsymbol c_k),
\pi_{\mathcal C}
\bigr)
+
\frac1{\sqrt N}
\int_{\mathcal C}
C(\boldsymbol c)
\,
\mathcal L(\boldsymbol c_k)
(d\boldsymbol c).
\end{align}

\end{enumerate}

\end{theorem}

\begin{proof}

By Lemma~\ref{lem:3.1}, the estimator
\(
\widehat{\tilde\pi}_{\mathcal C}
(\boldsymbol c,U)
\)
is nonnegative and unbiased
\[
\mathbb E_{U\sim\mu_{N,\boldsymbol c}}
\Bigl[
\widehat{\tilde\pi}_{\mathcal C}
(\boldsymbol c,U)
\Bigr]
=
\tilde\pi_{\mathcal C}(\boldsymbol c).
\]
Therefore, the standard pseudo-marginal construction of
\citet{J30}
implies that the kernel
\(
\widetilde P
\)
satisfies detailed balance with respect to
\[
\Pi(d\boldsymbol c,dU)
\propto
\tilde\pi_{\mathcal C}(\boldsymbol c)
\,
\mu_{N,\boldsymbol c}(dU).
\]
Hence,
\(
\Pi
\)
is invariant, establishing part~(i).

Since the proposal density
\(
q(\boldsymbol c'\mid\boldsymbol c)
\)
is strictly positive, the induced Metropolis--Hastings chain on
\(
\mathcal C
\)
is irreducible and aperiodic by assumption. Theorem~1 of
\citet{J30}
then implies that the augmented pseudo-marginal chain is also irreducible and aperiodic. Consequently,
\(
\widetilde P
\)
is Harris ergodic, which proves part~(ii).

Let
\(
\widehat\pi_N(\cdot\mid\boldsymbol c)
\)
denote the empirical reconstruction distribution obtained from the importance-resampling step. Conditioning on
\(
\boldsymbol c_k
\)
and applying Lemma~\ref{lem:3.4} gives
\[
d_{\mathrm{TV}}
\Bigl(
\widehat\pi_N(\cdot\mid\boldsymbol c_k),
\pi(\cdot\mid\boldsymbol c_k)
\Bigr)
\le
\frac{
C(\boldsymbol c_k)
}{\sqrt N}.
\]
Using the decomposition
\(
\pi(d\boldsymbol\theta)
=
\pi(d\boldsymbol\theta\mid\boldsymbol c)
\,
\pi_{\mathcal C}(d\boldsymbol c),
\)
together with the triangle inequality for total variation distance, we have
\begin{align*}
d_{\mathrm{TV}}
\bigl(
\mathcal L(\boldsymbol\theta_k),
\pi
\bigr)
\le
d_{\mathrm{TV}}
\bigl(
\mathcal L(\boldsymbol c_k),
\pi_{\mathcal C}
\bigr)
+
\int_{\mathcal C}
d_{\mathrm{TV}}
\Bigl(
\widehat\pi_N(\cdot\mid\boldsymbol c),
\pi(\cdot\mid\boldsymbol c)
\Bigr)
\,
\mathcal L(\boldsymbol c_k)
(d\boldsymbol c).
\end{align*}
Applying the conditional reconstruction bound gives
\[
d_{\mathrm{TV}}
\bigl(
\mathcal L(\boldsymbol\theta_k),
\pi
\bigr)
\le
d_{\mathrm{TV}}
\bigl(
\mathcal L(\boldsymbol c_k),
\pi_{\mathcal C}
\bigr)
+
\frac1{\sqrt N}
\int_{\mathcal C}
C(\boldsymbol c)
\,
\mathcal L(\boldsymbol c_k)
(d\boldsymbol c),
\]
which establishes part~(iii).

\end{proof}
For general rational identifiable mappings
\(
\boldsymbol{\xi}
\),
Jacobian degeneracies may occur on lower-dimensional algebraic subsets of the parameter space. In such cases, the constants
\(
C(\boldsymbol c)
\)
appearing in the conditional reconstruction bounds may become arbitrarily large, and uniform minorisation conditions are generally unavailable. Consequently, one cannot expect uniform ergodicity of the pseudo-marginal chain or uniform variance bounds for the estimator
\(
\widehat{\tilde\pi}_{\mathcal C}
\). Theorem~\ref{thm:weak-tv}(iii) presents the general error decomposition
\[
d_{\mathrm{TV}}
\Bigl(
\mathcal L(\boldsymbol\theta_k),
\pi
\Bigr)
\le
r(k)
+
\frac{C_k}{\sqrt N},
\]
where
\(
r(k)
=
d_{\mathrm{TV}}
\Bigl(
\mathcal L(\boldsymbol c_k),
\pi_{\mathcal C}
\Bigr)
\)
denotes the convergence error of the marginal Markov chain on the identifiable coordinates, and
\(
C_k
=
\int_{\mathcal C}
C(\boldsymbol c)
\,
\mathcal L(\boldsymbol c_k)
(d\boldsymbol c).
\)
This decomposition separates two distinct sources of error. The first term,
\(
r(k)
\),
is the mixing error of the marginal Markov chain on
\(
\mathcal C
\)
and depends on the ergodic properties of the pseudo-marginal kernel. The second term,
\(
O(N^{-1/2})
\),
is the reconstruction error from the finite number of conditional samples used on each fibre
\(
\mathcal M_{\boldsymbol c}
\),
and decreases at the standard Monte Carlo rate.
The two errors are controlled independently.
The iteration number
\(
k
\)
governs convergence of the marginal Markov chain, whereas the number of conditional samples
\(
N
\)
governs the accuracy of reconstruction on the fibres. 

Stronger convergence results can be obtained when the identifiable map is uniformly regular. Suppose that
\[
\sigma_{\min}
\bigl(
\nabla\boldsymbol{\xi}(\boldsymbol{\theta})
\bigr)
\ge
\delta>0,
\qquad
\forall\,
\boldsymbol{\theta}\in\Theta,
\]
and that the fibres
\(
\mathcal M_{\boldsymbol c}
\)
form a smooth compact foliation of
\(
\Theta
\).
Then the coarea factor and conditional densities are uniformly bounded, and there exists a constant
\(
C_\delta<\infty
\)
such that
\(
\sup_{\boldsymbol c\in\mathcal C}
C(\boldsymbol c)
\le
C_\delta .
\)
If, in addition, the proposal density
\(
q(\boldsymbol c'\mid\boldsymbol c)
\)
is bounded above and below by positive constants on the compact space
\(
\mathcal C
\),
then the marginal Metropolis--Hastings chain on
\(
\mathcal C
\)
is uniformly ergodic. Hence there exist constants
\(
M>0
\)
and
\(
\rho\in(0,1)
\)
such that
\(
r(k)
\le
M\rho^k .
\)
Combining this estimate with Theorem~\ref{thm:weak-tv}(iii) gives
\[
d_{\mathrm{TV}}
\Bigl(
\mathcal L(\boldsymbol\theta_k),
\pi
\Bigr)
\le
M\rho^k
+
\frac{C_\delta}{\sqrt N}.
\]
To achieve overall accuracy
\(
\varepsilon
\) in the uniformly regular setting, we can balance the two errors by choosing
\(
\rho^k
\asymp
N^{-1/2}
\asymp
\varepsilon.
\)
This yields
\[
k
=
O\!\left(
\log(\varepsilon^{-1})
\right),
\qquad
N
=
O\!\left(
\varepsilon^{-2}
\right).
\]
Although the global uniform regularity condition may fail in practice, the set of singular points has Lebesgue measure zero under Assumption~\ref{ass:rational}(i), and
\(
\boldsymbol{\xi}
\)
The values of \(k\) and \(N\) derived under the assumption of geometric convergence provide a useful practical approximation for tuning the algorithm. When an uninformative prior is employed and the posterior distribution is entirely governed by the likelihood, the reconstruction error can be ignored. In this setting, the optimal algebraic manifold \(\mathcal{M}_{\boldsymbol{c}}\) already identifies the manifold that best explains the observed data. All points on the optimized manifold are observationally equivalent from the perspective of parameter inference. 

Compared with the identifiability-aware geometric MCMC method, both algorithms require repeated solutions of the algebraic system
\(
\boldsymbol{\xi}(\boldsymbol{\theta})=\boldsymbol{c}.
\)
The key difference lies in the choice of solver. Geometric MCMC employs local solvers to generate proposals along identifiable level sets, whereas the pseudo-marginal approach relies on global solvers to approximate the marginal likelihood. Since the computational cost of global methods like homotopy continuation increases rapidly with the dimension of the parameter space, the additional cost may offset the benefit obtained from reducing the dimension of the inference problem. Therefore, the identifiability-aware pseudo-marginal MCMC approach is most attractive for models with a relatively small number of parameters, where global solution of the algebraic system remains computationally efficient.
\begin{algorithm}[htbp]
\caption{Identifiability Aware Pseudo-Marginal MCMC}
\label{alg:IAPMCMC}
\begin{algorithmic}[1]
\Require Initial value $\boldsymbol{c}_0\in\mathcal{C}$, proposal density $q$ and Algorithm \ref{alg:coord_partition}.

\State Compute pseudo-marginal estimator
\(
\widehat{\pi}_{\mathcal C}(\boldsymbol c_0)
\) using Algorithm \ref{alg:coord_partition}.

\For{$k=0,1,2,\dots, K$}

\State Propose
\(
\boldsymbol c'
\sim
q(\boldsymbol c'\mid\boldsymbol c_k).
\)

\State Compute pseudo-marginal estimator
\(
\widehat{\pi}_{\mathcal C}(\boldsymbol c').
\) using Algorithm \ref{alg:coord_partition}.

\State Compute the Metropolis--Hastings acceptance probability
\(
\alpha
\) in equation \eqref{eq:alpha}.

\State Sample
\(
u\sim\mathrm{Unif}(0,1)
\).

\If{$u<\alpha$}
\State Set
\(
(\boldsymbol c_{k+1},U_{k+1})
=
(\boldsymbol c',U').
\)
\Else
\State Set
\(
(\boldsymbol c_{k+1},U_{k+1})
=
(\boldsymbol c_k,U_k).
\)
\EndIf

\State Reconstruct representative parameter sample
\[
\boldsymbol{\theta}_{k+1}
\sim
\sum_{i=1}^N
\bar w_i
\,
\delta_{\boldsymbol{\theta}^{(i)}_{k+1}}.
\]

\EndFor
\end{algorithmic}
\end{algorithm}

\section{Case Studies}\label{sec:case}

In this section, we present two case studies to show the performance of the proposed identifiability-aware MCMC algorithms. Both examples arise from compartmental ODE models for public health applications. The models considered here have known structural non-identifiabilities, making them suitable benchmarks for evaluating the proposed methods. The MCMC algorithms compared in this study include the standard random-walk MCMC (baseline), the identifiability-aware geometric MCMC (Algorithm \ref{alg:IAGMCMC}), and the identifiability-aware pseudo-marginal MCMC (Algorithm \ref{alg:IAPMCMC}). To ensure a fair comparison, all numerical experiments are conducted under identical settings, including the same initial parameter values and the same number of MCMC iterations.

\subsection{SI Model}\label{sec:4.1}

We first consider the SI model introduced in Example~\eqref{eq:SI}. The state variables $(S(t),I(t))\in\mathbb{R}^2$ represent the susceptible and infectious populations. The parameter vector is
\(
\boldsymbol{\theta}=(\beta,\rho,\gamma,I_0)\in\mathbb{R}^4,
\)
where $\beta$ denotes the transmission rate, $\rho$ the case-detection rate, $\gamma$ the recovery rate, and $I_0$ the initial number of infectious individuals. We assume that the total population size $N$ is known, and the initial susceptible population is given by $S_0=N-I_0$.
Observational data are collected through recording daily detected case number, and the corresponding model output is a time series
\(
\bigl\{(t_i,\rho I(t_i))\bigr\}_{i=1}^{T}.
\)
We assume independent Gaussian observation noises and model the observations as
\begin{equation}\label{eq:gnoise}
Y_i \sim\mathcal{N}\left(\rho I(t_i;\boldsymbol{\theta}),,\sigma^2\right),
\qquad i=1,\ldots,T,
\end{equation}
where $\sigma>0$ is the observation noise standard deviation. The objective is to sample the posterior distribution of \(\boldsymbol{\theta}\) and obtain reliable parameter estimates and predictive trajectories.
This model is structurally non-identifiable. Since the total population size \(N=S+I\) is known, we may eliminate \(S\) and write
\[
I'=(\beta N-\gamma)I-\beta I^2.
\]
Using the observation equation \(y=\rho I\), the state variable \(I\) can be eliminated to obtain the input--output equation
\[
y'=(\beta N-\gamma)y-\frac{\beta}{\rho}y^2.
\]
It follows that the observable dynamics depend on the parameters only through the combinations
\(
\boldsymbol{\xi}(\boldsymbol{\theta})
=
(
\beta N-\gamma,
\beta/\rho,
\rho I_0
).
\)
The parameter space is partitioned into non-identifiable manifolds
\(
\mathcal M_{\boldsymbol c}
=
\left\{
\boldsymbol{\theta}:
\boldsymbol{\xi}(\boldsymbol{\theta})
=
\boldsymbol c
\right\},
\)
along which the model output remains invariant. For any fixed
\(
\boldsymbol{c}=(c_1,c_2,c_3)\in \mathbb{R}^3
\),
the manifold
\[
\mathcal M_{\boldsymbol{c}}
=
\left\{
(\beta,\rho,\gamma,I_0):
\beta N-\gamma=c_1,\;
\frac{\beta}{\rho}=c_2,\;
\rho I_0=c_3
\right\}
\]
is one-dimensional. It admits the parameterization
\begin{equation}\label{eq:parms}
\beta=c_2\rho,
\qquad
\gamma=c_2\rho N-c_1, \qquad I_0=\frac{c_3}{\rho}
\end{equation}
with \(\rho>0\) free. This shows that $\rho$ can be chosen as an independent parameter and the remaining parameters are determined by the identifiable combination values $(c_1,c_2,c_3)$, as we discussed in Section \ref{sec:3.3}. 
Synthetic observations are generated from the SI model using the parameter values
\[
\beta^{\ast}=2.5\times10^{-4},
\qquad
\rho^{\ast}=0.25,
\qquad
\gamma^{\ast}=0.6,
\qquad
I_0^{\ast}=40,
\]
with total population size \(N=10^4\). These values produce a clear epidemic outbreak and a nontrivial case trajectory. Observations are collected at the first ten integer time points and are contaminated with independent Gaussian noise according to
\[
Y_i \sim \mathcal N\!\left(y(t_i;\boldsymbol{\theta}^{\ast}),\sigma^2\right),
\]
where \(\sigma=5\).
For Bayesian inference, independent uniform prior distributions are assigned to all unknown parameters
\[
\beta \sim U(10^{-5},10^{-3}),
\qquad
\rho \sim U(0.01,1),
\]
\[
\gamma \sim U(0.01,1),
\qquad
I_0 \sim U(1,500).
\]
The prior support contains the true parameter values and is intentionally chosen to be sufficiently broad to expose the effects of structural non-identifiability.

We compare three MCMC algorithms for posterior inference of the parameter vector \(\boldsymbol{\theta}\): a standard random-walk Metropolis--Hastings sampler, the identifiability-aware geometric MCMC algorithm (Algorithm~\ref{alg:IAGMCMC}), and the identifiability-aware pseudo-marginal MCMC algorithm (Algorithm~\ref{alg:IAPMCMC}). All methods are run for \(100{,}000\) iterations and initialized from the same randomly generated parameter vector drawn from the prior distribution. In addition, all algorithms use the same observation model, prior distributions, and burn-in period of \(1,000\) iterations to ensure a fair comparison.
For the baseline random-walk sampler, proposals are constructed on the logarithmic parameter scale. Specifically, parameters are transformed using a logarithmic mapping before sampling and subsequently transformed back to the original scale through exponentiation. This reparameterization alleviates the large differences in parameter magnitudes and improves numerical efficiency. The proposal kernel is Gaussian with covariance
\(
0.05 I.
\)
For Algorithm~\ref{alg:IAGMCMC}, the same logarithmic transformation is applied. The teleportation step is implemented using the RATTLE integrator described in Algorithm~\ref{alg:rattle}, with step size
\(
\epsilon = 0.005
\)
and \(20\) RATTLE steps per teleportation move. For the transition kernel defined in \eqref{eq:qnormal}, we set
\[
\sigma_N = 0.8,
\qquad
\sigma_T = 0.2,
\]
so that the majority of the proposal variance is concentrated in directions normal to the non-identifiable manifold.
For Algorithm~\ref{alg:IAPMCMC}, the pseudo-marginal estimator is constructed using \(128\) auxiliary samples. The proposal density \(g_{\boldsymbol c}\) is chosen to be uniform over the same interval as the prior distribution of the independent parameter \(\rho\). Since the SI model admits an analytical parameterization of \(\beta\), \(\gamma\), and \(I_0\) in terms of \(\rho\), conditional samples on the non-identifiable manifold can be generated directly, and no global polynomial solver is required.

The algebraic relationships in \eqref{eq:parms} imply that the non-identifiable manifold has a curved inverse relationship when projected onto the \((I_0,\rho)\)-plane. 
Figure~\ref{fig2a} displays a few projected samples generated by the baseline random-walk MCMC algorithm. The chain remains confined to a small region of the non-identifiable manifold and exhibits only local exploration. Since proposals are generated through small isotropic perturbations, movement along the curved manifold is slow, leading to strong serial correlation and poor mixing.
Figure~\ref{fig2b} illustrates the sampling mechanism of Algorithm~\ref{alg:IAGMCMC}. The RATTLE-based teleportation step enables large moves along the non-identifiable manifold. These manifold-preserving transitions are complemented by local moves between nearby manifolds, allowing the chain to explore both identifiable and non-identifiable directions more efficiently. 
Figure~\ref{fig2c} shows the sampling mechanism of Algorithm~\ref{alg:IAPMCMC}. The algorithm first samples on the space of identifiable combinations and then reconstructs parameter values from multiple conditional samples on the corresponding observationally equivalent manifold. 
The improvements in algorithm convergence and sampling efficiency are further reflected in the trace plots shown in Figures~\ref{fig1} and \ref{fig3}. Both identifiability-aware algorithms exhibit faster mixing and more rapid exploration of the posterior distribution than the baseline random-walk sampler.

\begin{figure}[H]
    \centering

    \begin{subfigure}[t]{0.32\textwidth}
        \centering
        \includegraphics[width=\linewidth]{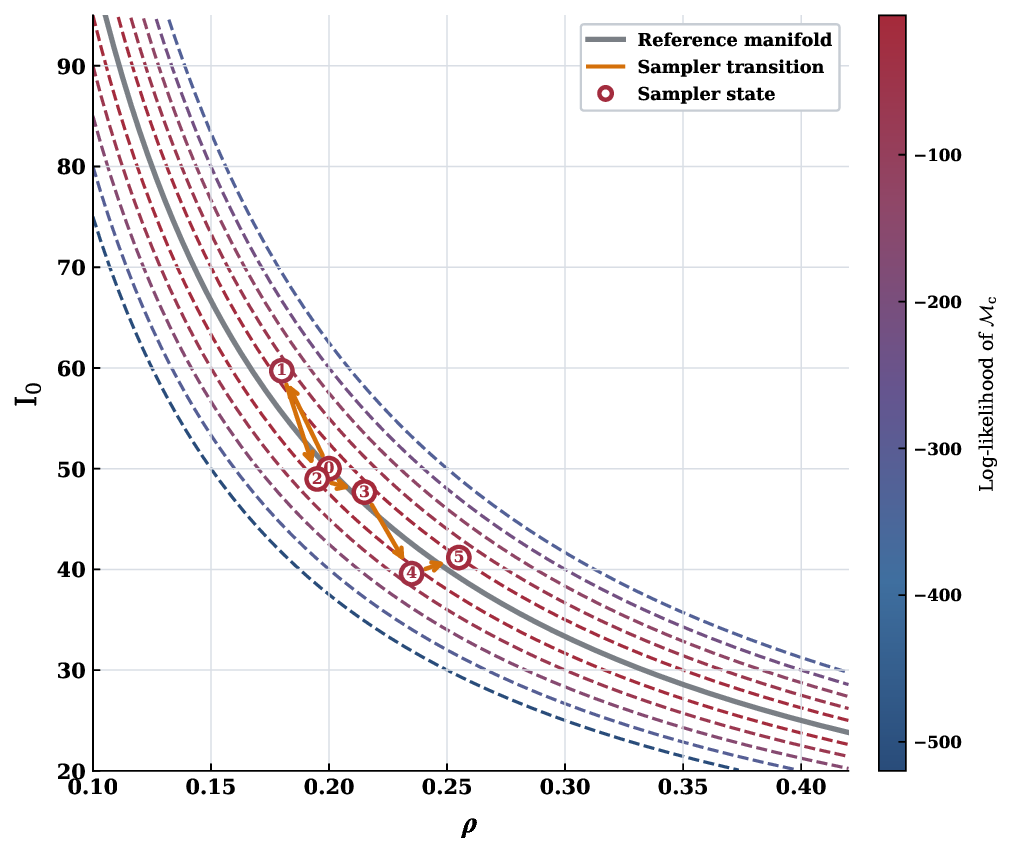}
        \subcaption{Random-walk MCMC.}\label{fig2a}
    \end{subfigure}\hfill
    \begin{subfigure}[t]{0.32\textwidth}
        \centering
        \includegraphics[width=\linewidth]{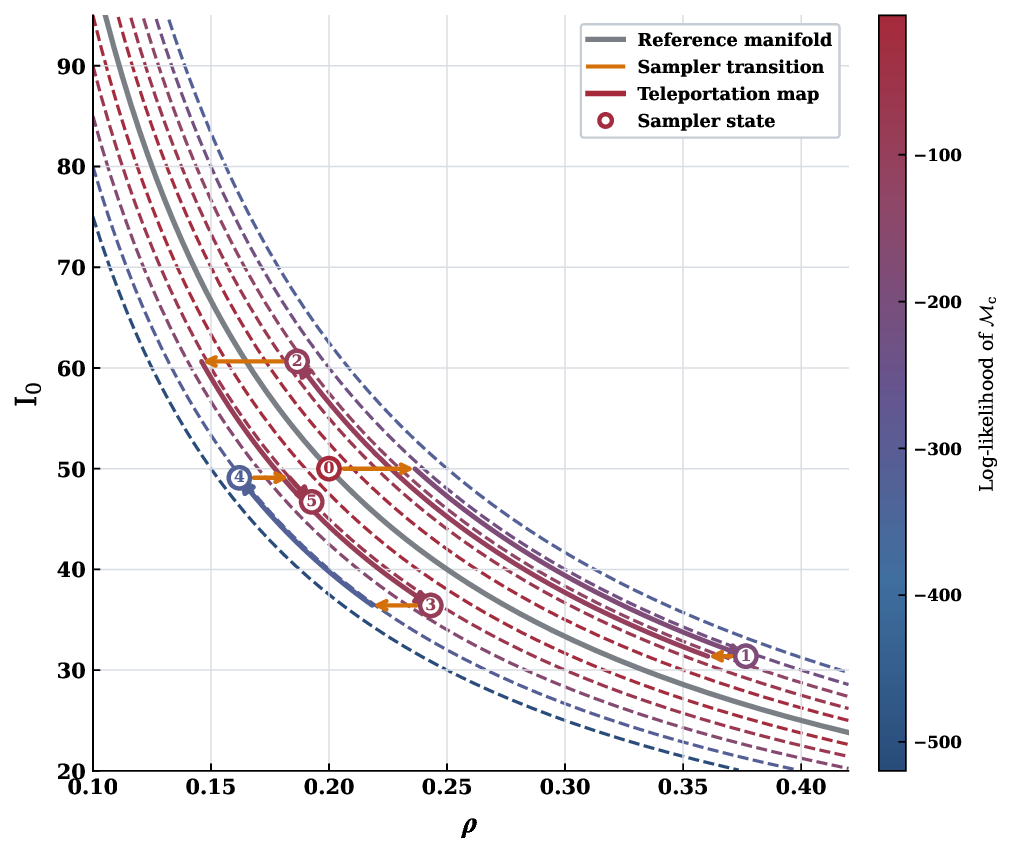}
        \subcaption{Algorithm \ref{alg:IAGMCMC}.}\label{fig2b}
    \end{subfigure}\hfill
    \begin{subfigure}[t]{0.32\textwidth}
        \centering
        \includegraphics[width=\linewidth]{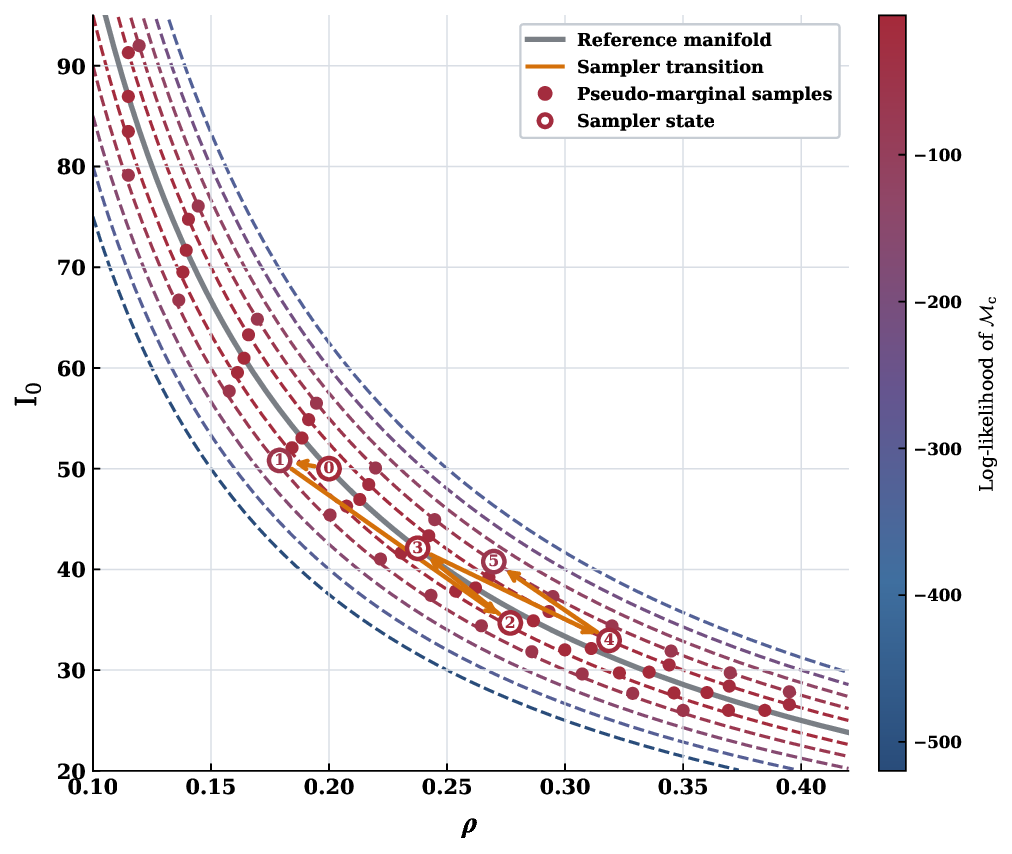}
        \subcaption{Algorithm \ref{alg:IAPMCMC}.}\label{fig2c}
    \end{subfigure}
    \caption{Comparison of the sampling mechanisms of different MCMC algorithms.}
    \label{fig2}
\end{figure}

\begin{figure}[H]
    \centering

    \begin{subfigure}[t]{0.45\textwidth}
        \centering
        \includegraphics[width=\linewidth]{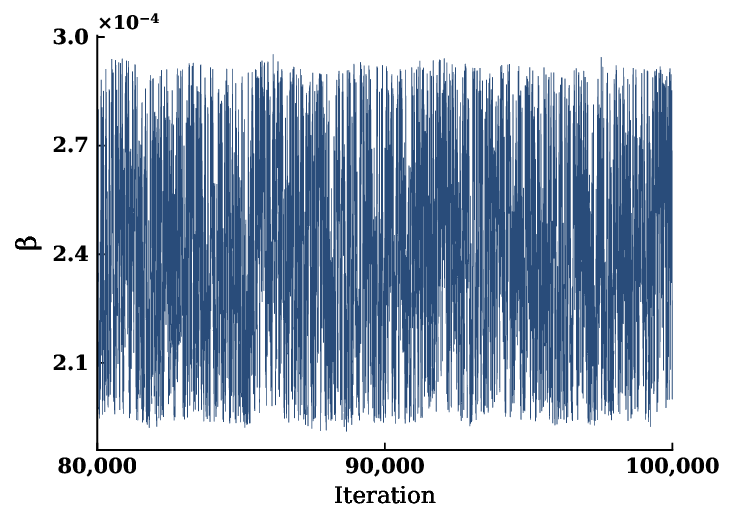}
        \subcaption{Algorithm \ref{alg:IAGMCMC}: trace plot for $\beta$.}
    \end{subfigure}\hfill
    \begin{subfigure}[t]{0.45\textwidth}
        \centering
        \includegraphics[width=\linewidth]{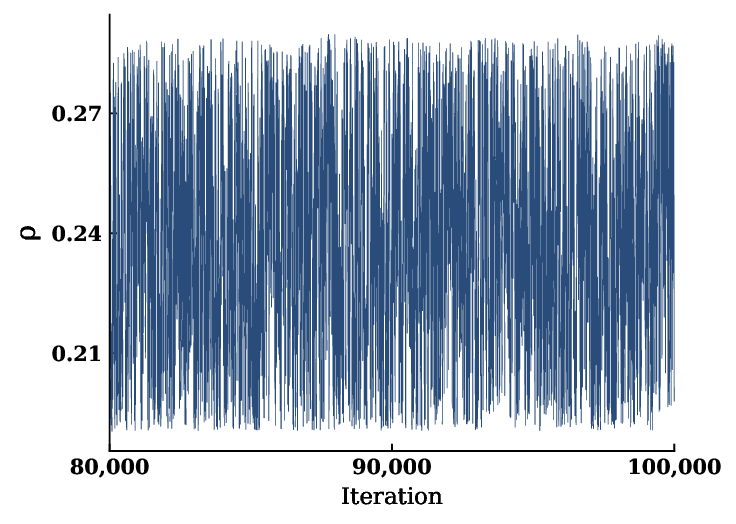}
        \subcaption{Algorithm \ref{alg:IAGMCMC}: trace plot for $\rho$.}
    \end{subfigure}

    \vspace{0.6em}

    \begin{subfigure}[t]{0.45\textwidth}
        \centering
        \includegraphics[width=\linewidth]{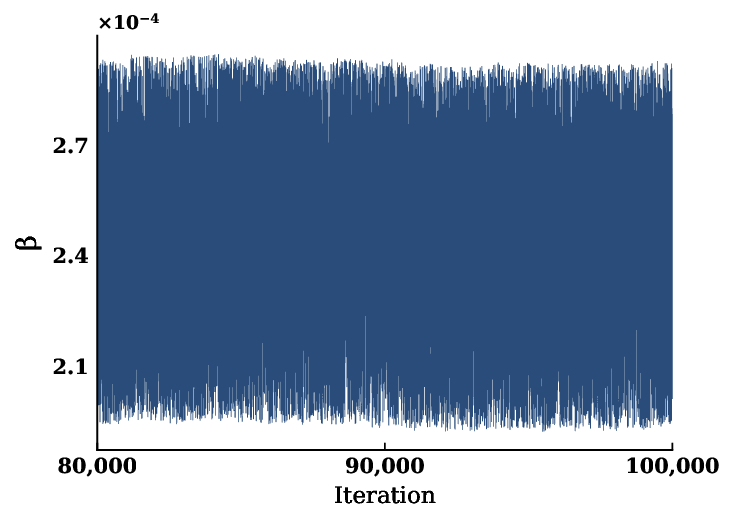}
        \subcaption{Algorithm \ref{alg:IAPMCMC}: trace plot for $\beta$.}
    \end{subfigure}\hfill
    \begin{subfigure}[t]{0.45\textwidth}
        \centering
        \includegraphics[width=\linewidth]{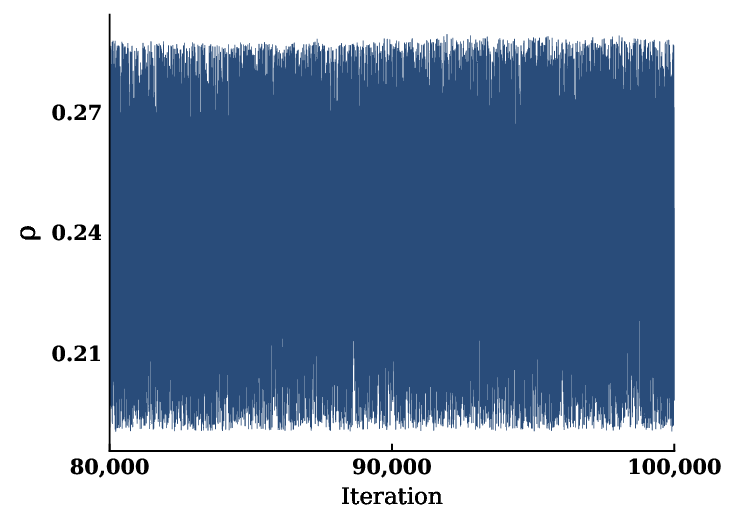}
        \subcaption{Algorithm \ref{alg:IAPMCMC}: trace plot for $\rho$.}
    \end{subfigure}

    \caption{Trace plots of the proposed identifiability-aware geometric MCMC (Algorithm~\ref{alg:IAGMCMC}) and identifiability-aware pseudo-marginal MCMC (Algorithm~\ref{alg:IAPMCMC}) for the structurally non-identifiable SI model.}
    \label{fig3}
\end{figure}

To quantitatively assess sampling efficiency, Table~\ref{tab:si_convergence} reports effective sample sizes (ESS), ESS per second (ESS/sec), integrated autocorrelation times (IACT), the ratio of Monte Carlo standard error to posterior standard deviation (MCSE/SD), and split-\(\widehat{R}\) diagnostics. The baseline random-walk sampler produces ESS values between 64 and 76 despite 100,000 iterations, with IACT values exceeding 1300 for all parameters and \(\widehat{R}\) values noticeably above one. These diagnostics indicate substantial autocorrelation and incomplete mixing.
In contrast, Algorithm~\ref{alg:IAGMCMC} increases the ESS by approximately one order of magnitude, yielding ESS values between 288 and 633 and reducing the IACT to approximately \(160\)--\(350\). The acceptance rate increases from \(27.93\%\) to \(53.67\%\), while the MCSE is reduced by a factor of approximately two across all parameters.
The largest gains are obtained by Algorithm~\ref{alg:IAPMCMC}. Effective sample sizes exceed \(4500\) for every parameter and reach more than \(5700\) for \(\beta\), \(\rho\), and \(\gamma\), representing improvements of roughly two orders of magnitude over the baseline sampler. The corresponding IACT values are reduced to approximately \(18\)--\(22\), indicating near-independent sampling behavior. Furthermore, all \(\widehat{R}\) values are essentially equal to one and the MCSE remains below \(1.5\%\) of the posterior standard deviation, providing strong evidence of convergence.

\begin{table}[ht]
\centering
\begin{tabular}{c c c c c c c}
\toprule
Method & Parameter & ESS & ESS/sec & IACT & MCSE/SD (\%) & $\widehat{R}$ \\
\midrule

\multirow{4}{*}{\shortstack{RW-MCMC\\$\alpha=27.93\%$}}
& $\beta$  & 66.68 & 0.67 & 1499.77 & 12.25 & 1.13 \\
& $\rho$   & 76.26 & 0.76 & 1311.25 & 11.45 & 1.11 \\
& $\gamma$ & 64.30 & 0.64 & 1555.09 & 12.47 & 1.14 \\
& $I_0$    & 65.94 & 0.66 & 1516.51 & 12.31 & 1.13 \\
\midrule

\multirow{4}{*}{\shortstack{Algorithm \ref{alg:IAGMCMC}\\$\alpha=53.67\%$}}
& $\beta$  & 592.70 & 4.94 & 168.72 & 4.11 & 1.04 \\
& $\rho$   & 633.23 & 5.28 & 157.92 & 3.97 & 1.03 \\
& $\gamma$ & 624.01 & 5.20 & 160.25 & 4.00 & 1.03 \\
& $I_0$    & 288.23 & 2.40 & 346.95 & 5.89 & 1.05 \\
\midrule

\multirow{4}{*}{\shortstack{Algorithm \ref{alg:IAPMCMC}\\$\alpha=32.75\%$}}
& $\beta$  & 5587.33 & 27.94 & 17.90 & 1.34 & 1.01 \\
& $\rho$   & 5725.01 & 28.63 & 17.47 & 1.32 & 1.00 \\
& $\gamma$ & 5696.37 & 28.48 & 17.56 & 1.33 & 1.00 \\
& $I_0$    & 4561.61 & 22.81 & 21.92 & 1.48 & 1.01 \\
\bottomrule
\end{tabular}
\caption{Convergence diagnostics for the SI model based on 100,000 MCMC iterations.}
\label{tab:si_convergence}
\end{table}

Overall, the numerical results demonstrate that explicitly incorporating structural identifiability information can dramatically improve posterior sampling. The geometric sampler of Algorithm~\ref{alg:IAGMCMC} effectively exploits the manifold structure to accelerate exploration, while the pseudo-marginal formulation of Algorithm~\ref{alg:IAPMCMC} achieves the highest sampling efficiency by performing inference directly on the identifiable combinations and reconstructing posterior samples on the non-identifiable manifolds. For this SI model, Algorithm~\ref{alg:IAPMCMC} provides the best overall balance of convergence, mixing, and computational efficiency.

\subsection{HIV Model}

We next consider an HIV infection model from \citet{J37}. The state variables
\(
(T(t),T^{\ast}(t),V(t))\in\mathbb{R}^{3}
\)
represent the concentrations of susceptible target cells, infected cells, and free virus particles, respectively. The model is given by
\begin{equation}\label{eq:hiv}
\begin{cases}
\displaystyle \frac{dT}{dt} = \lambda-\rho T-\beta TV,\\[1ex]
\displaystyle \frac{dT^{\ast}}{dt} = \beta TV-\delta T^{\ast},\\[1ex]
\displaystyle \frac{dV}{dt} = N\delta T^{\ast}-cV,
\end{cases}
\end{equation}
with observation function
\(
y(t)=V(t).
\)
The parameter vector is
\(
\boldsymbol{\theta}
=
(\beta,\rho,\delta,c,\lambda,N)\in\mathbb{R}^{6},
\)
where \(\beta\) denotes the infection rate, \(\rho\) the natural death rate of susceptible target cells, \(\delta\) the death rate of infected cells, \(c\) the viral clearance rate, \(\lambda\) the source rate of susceptible target cells, and \(N\) the average number of virions produced by an infected cell during its lifetime. Observational data are collected through the viral load $V(t)$
so that the model output consists of the time series
\(
\bigl\{(t_i,V(t_i))\bigr\}_{i=1}^{T}.
\)
Similar to previous section, we assume independent Gaussian observation errors and model the data observations according to 
$$Y_i\sim \mathcal{N}(V(t_i,\boldsymbol{\theta}),\sigma^2), \qquad i=1,\cdots, T.$$
The objective is to sample the posterior distribution of \(\boldsymbol{\theta}\) and obtain reliable parameter estimates.
Structural identifiability analysis of this model have been carried out by \citet{J38}. Using differential-algebraic method, it can be shown that the model is structurally non-identifiable with one non-identifiable degree of freedom. A maximal set of algebraically independent identifiable parameter combinations is given by
\[
\boldsymbol{\xi}(\boldsymbol{\theta})
=
\left(
\frac{\lambda N}{c},
\, c,
\, \rho,
\, \beta,
\, \delta
\right).
\]
The parameter space is partitioned by non-identifiable manifolds
\(
\mathcal M_{\boldsymbol c}
=
\left\{
\boldsymbol{\theta}:
\boldsymbol{\xi}(\boldsymbol{\theta})
=
\boldsymbol c
\right\}
\),
which are one-dimensional. Along each manifold \(\mathcal M_{\boldsymbol c}\), the model output remains invariant. 
Synthetic observations are generated from \eqref{eq:hiv} using the parameter values listed in Table~\ref{tab:HIV_param}. These values produce a realistic viral load trajectory exhibiting an initial transient phase followed by stabilization, which is representative of the dynamics observed in HIV infection models. Observations are collected daily over a period of 30 days, 
and are corrupted according to the Gaussian observation model \eqref{eq:gnoise} with observation noise standard deviation \(\sigma=50\). The resulting dataset is used as input for Bayesian inference.
Independent uniform priors are assigned to all unknown parameters, with prior bounds given in Table~\ref{tab:HIV_param}. The prior ranges contain the ground-truth parameter values and are sufficiently broad to illustrate the effects of structural non-identifiability.
\begin{table}[ht]
\centering
\begin{tabular}{llll}
\hline
Parameter & Description & True value & Prior distribution \\
\hline
$\beta$ &
Infection rate &
$2.4\times10^{-5}$ &
$U(10^{-6},10^{-3})$ \\

$\rho$ &
Death rate of susceptible target cells &
$0.01$ &
$U(10^{-3},0.1)$ \\

$\delta$ &
Death rate of infected cells &
$0.5$ &
$U(0.01,2)$ \\

$c$ &
Viral clearance rate &
$3$ &
$U(

1,10)$ \\

$\lambda$ &
Source rate of susceptible target cells &
$10$ &
$U(1,100)$ \\

$N$ &
Virions produced per infected cell &
$1000$ &
$U(100,5000)$ \\
\hline
\end{tabular}
\caption{Parameter definitions, ground-truth values, and prior distributions for HIV case study.}
\label{tab:HIV_param}
\end{table}

Similar to the SI case study, all parameters are sampled on the logarithmic scale to account for differences in parameter magnitudes. All algorithms are initialized from the same randomly generated parameter vector drawn from the prior distribution and run for \(100{,}000\) iterations with \(1{,}000\) burn-in iterations. For the baseline random-walk MCMC algorithm, a Gaussian proposal with covariance matrix \(0.01I\) is employed, while the remaining hyperparameters are chosen identically to those used in the SI example.

The convergence diagnostics reported in Table~\ref{tab:hiv_convergence} show a clear distinction between structurally identifiable and non-identifiable parameters. For the identifiable parameters \(\beta\), \(\rho\), \(\delta\), and \(c\), all three algorithms achieve comparable performance, with only moderate improvements obtained by the proposed methods. In contrast, substantial differences are observed for the structurally non-identifiable parameters \(\lambda\) and \(N\). The baseline random-walk sampler exhibits strong autocorrelation and poor mixing, whereas Algorithm~\ref{alg:IAGMCMC} improves exploration of the non-identifiable manifold. The largest gains are achieved by Algorithm~\ref{alg:IAPMCMC}, which increases the ESS of \(\lambda\) and \(N\) by more than an order of magnitude while dramatically reducing both IACT and MCSE. These findings are further supported by the autocorrelation functions shown in Figure~\ref{fig:acf_comparison}. For the identifiable parameter \(\beta\), all methods display similar autocorrelation decay, whereas for the non-identifiable parameter \(\lambda\), the proposed algorithms decorrelate substantially faster than the baseline sampler, with Algorithm~\ref{alg:IAPMCMC} exhibiting the most rapid decay. Overall, the results confirm that incorporating structural identifiability information primarily improves sampling efficiency for non-identifiable parameters.

\begin{table}[H]
\centering
\begin{tabular}{c c c c c c c}
\toprule
Method & Parameter & ESS & ESS/sec & IACT & MCSE/SD (\%) & $\widehat{R}$ \\
\midrule

\multirow{6}{*}{\shortstack{RW-MCMC\\$\alpha=18.22\%$}}
& $\beta$   & 511.18 & 5.11 & 195.63 & 4.42 & 1.04 \\
& $\rho$    & 413.92 & 4.14 & 241.59 & 4.92 & 1.05 \\
& $\delta$  & 424.15 & 4.24 & 235.76 & 4.86 & 1.05 \\
& $c$       & 313.87 & 3.14 & 318.60 & 5.64 & 1.06 \\
& $\lambda$ & 82.42  & 0.82 & 1213.31 & 11.01 & 1.10 \\
& $N$       & 85.19  & 0.85 & 1173.90 & 10.84 & 1.10 \\
\midrule

\multirow{6}{*}{\shortstack{Algorithm \ref{alg:IAGMCMC}\\$\alpha=30.92\%$}}
& $\beta$   & 783.63 & 6.53 & 127.61 & 3.57 & 1.03 \\
& $\rho$    & 590.65 & 4.92 & 169.30 & 4.11 & 1.04 \\
& $\delta$  & 580.52 & 4.84 & 172.26 & 4.15 & 1.04 \\
& $c$       & 573.21 & 4.78 & 174.46 & 4.18 & 1.04 \\
& $\lambda$ & 602.21 & 5.02 & 166.05 & 4.08 & 1.03 \\
& $N$       & 741.89 & 6.18 & 134.79 & 3.67 & 1.03 \\
\midrule

\multirow{6}{*}{\shortstack{Algorithm \ref{alg:IAPMCMC}\\$\alpha=49.43\%$}}
& $\beta$   & 608.52 & 3.04 & 164.33 & 4.05 & 1.03 \\
& $\rho$    & 609.19 & 3.05 & 164.16 & 4.05 & 1.03 \\
& $\delta$  & 617.36 & 3.09 & 161.98 & 4.02 & 1.03 \\
& $c$       & 612.36 & 3.06 & 163.30 & 4.04 & 1.03 \\
& $\lambda$ & 3381.88 & 16.91 & 29.57 & 1.72 & 1.01 \\
& $N$       & 4075.35 & 20.38 & 24.54 & 1.57 & 1.01 \\
\bottomrule
\end{tabular}
\caption{Convergence diagnostics for the HIV model based on 100,000 MCMC iterations.}
\label{tab:hiv_convergence}
\end{table}

We conclude the HIV case study by comparing the posterior inference results for a structurally identifiable parameter, a structurally non-identifiable parameter, and an identifiable parameter combination. These results highlight the fundamental differences between traditional MCMC methods and the proposed identifiability-aware algorithms.
Figure~\ref{fig5a} shows the posterior distribution of the identifiable parameter \(\beta\). All three algorithms recover the true parameter value accurately and produce similar posterior distributions, which is expected since \(\beta\) is structurally identifiable and can be informed directly by the data.
A significantly different behavior is observed for the non-identifiable parameter \(\lambda\) in Figure~\ref{fig5b}. The baseline random-walk MCMC algorithm assigns most posterior mass to values near \(20\), despite the true value being \(10\). This behavior is likely caused by poor mixing along the non-identifiable manifold, resulting in a posterior estimate that depends strongly on the initial state of the chain. In contrast, both Algorithm~\ref{alg:IAGMCMC} and Algorithm~\ref{alg:IAPMCMC} recover a nearly uniform posterior distribution across the admissible parameter range. This result is consistent with the structural non-identifiability of \(\lambda\): all values along the non-identifiable manifold generate observationally equivalent model outputs, and under the non-informative prior considered here, the posterior should remain diffuse over the entire manifold.
Although the proposed algorithms do not provide a more informative estimator for a structurally non-identifiable parameter such as \(\lambda\), this does not imply that no information about the parameter can be extracted from the data. Figure~\ref{fig5c} presents the posterior distribution of the identifiable combination \(\lambda c/N\). Both identifiability-aware algorithms accurately recover the true value and produce concentrated posterior distributions. In contrast, the baseline random-walk sampler does not recover the true value of the identifiable combination. These results demonstrate that the proposed methods correctly distinguish between identifiable and non-identifiable directions in the parameter space, allowing reliable inference for identifiable combinations when individual parameters cannot be uniquely determined. In fact, for structurally non-identifiable models, inference on identifiable combinations represents the maximum amount of information that can be recovered from the data without introducing additional prior information or observations \citep{J17}.

\begin{figure}[H]
    \centering

    \begin{subfigure}[t]{0.45\textwidth}
        \centering
        \includegraphics[width=\linewidth]{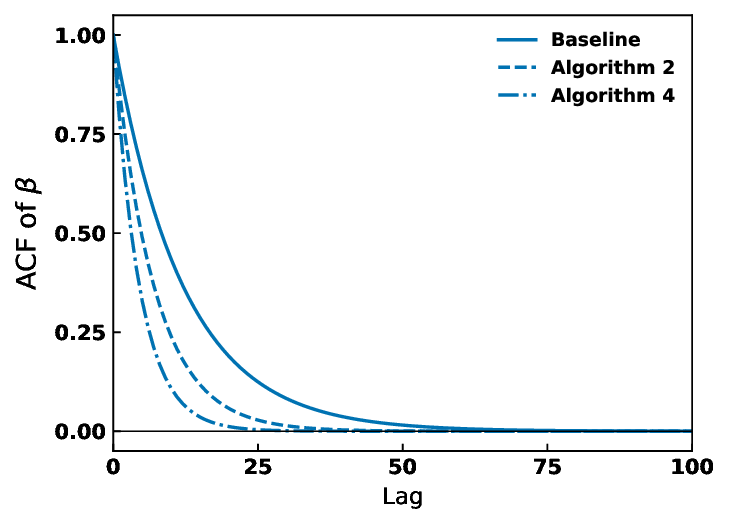}
        \subcaption{ACF for $\beta$.}
    \end{subfigure}\hfill
    \begin{subfigure}[t]{0.45\textwidth}
        \centering
        \includegraphics[width=\linewidth]{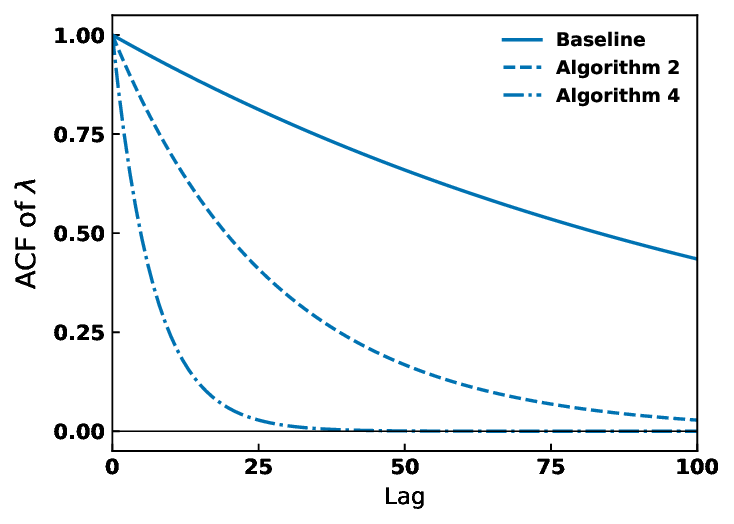}
        \subcaption{ACF for $\lambda$.}
    \end{subfigure}

    \caption{Comparison of autocorrelation functions (ACFs) of different MCMC methods for the identifiable parameter $\beta$ and the non-identifiable parameter $\lambda$.}
    \label{fig:acf_comparison}
\end{figure}

\begin{figure}[H]
    \centering

    \begin{subfigure}[t]{0.32\textwidth}
        \centering
        \includegraphics[width=\linewidth]{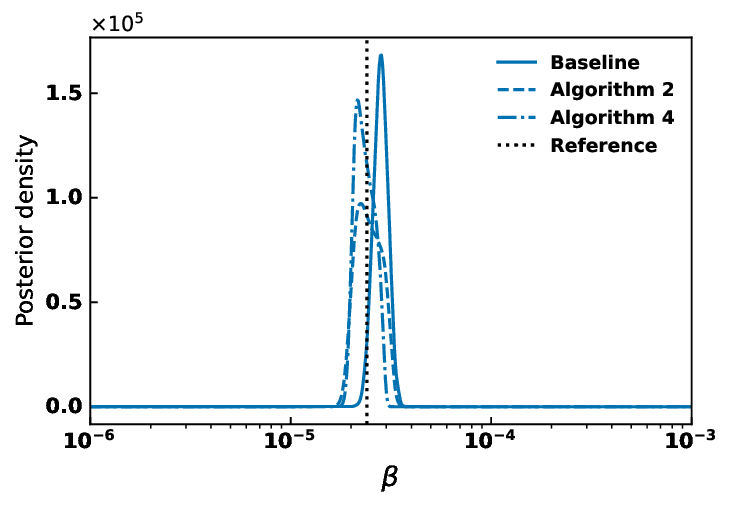}
        \subcaption{Posterior density of $\beta$.}\label{fig5a}
    \end{subfigure}\hfill
    \begin{subfigure}[t]{0.32\textwidth}
        \centering
        \includegraphics[width=\linewidth]{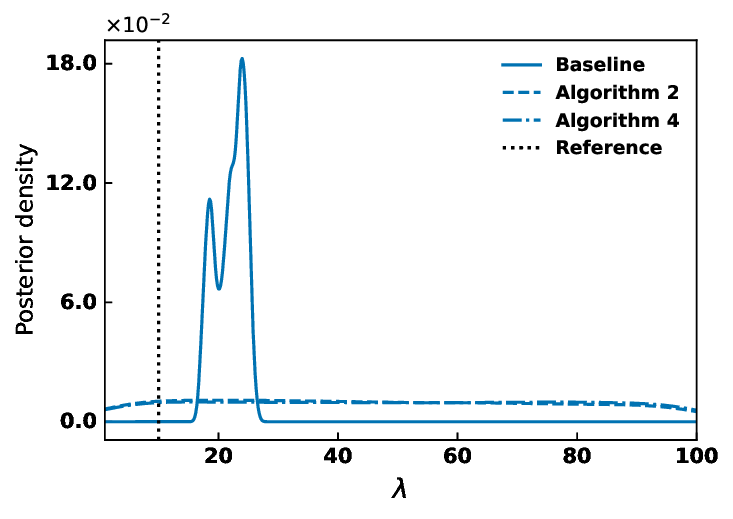}
        \subcaption{Posterior density of $\lambda$.}\label{fig5b}
    \end{subfigure}\hfill
    \begin{subfigure}[t]{0.32\textwidth}
        \centering
        \includegraphics[width=\linewidth]{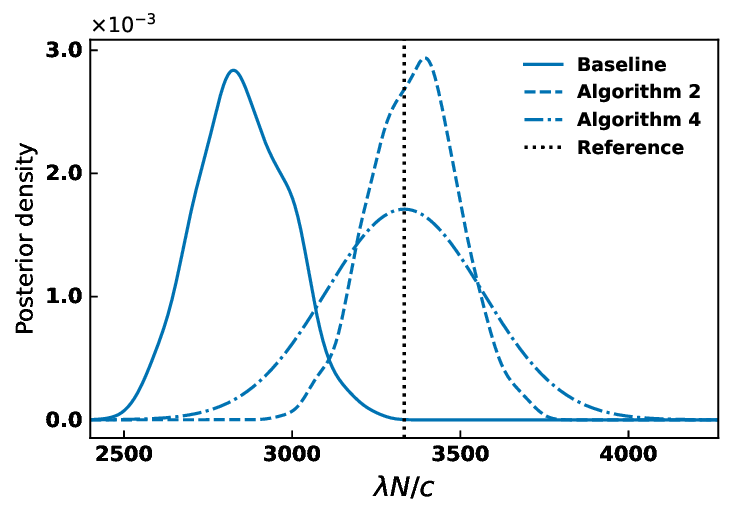}
        \subcaption{Posterior density of $\lambda c/N$.}\label{fig5c}
    \end{subfigure}

    \caption{Comparison of inference results obtained by different MCMC algorithms for the identifiable parameter $\beta$, the non-identifiable parameter $\lambda$, and the identifiable combination $\lambda c/N$.}
    \label{fig5}
\end{figure}

\section{Discussion} \label{sec:dis}

Structural non-identifiability presents a fundamental challenge for Bayesian inference in mechanistic models. When distinct parameter values produce identical observable outputs, posterior distributions become concentrated along lower-dimensional manifolds of observationally equivalent solutions. As a result, standard MCMC algorithms often exhibit slow mixing, strong autocorrelation, and poor convergence.

In this paper, we introduce two MCMC methodologies based on structural identifiability analysis results. The first, an identifiability-aware geometric MCMC algorithm, combines manifold-preserving teleportation moves with MCMC updates to improve exploration of non-identifiable manifolds. The second, an identifiability-aware pseudo-marginal MCMC algorithm, performs inference on the space of identifiable parameter combinations and subsequently reconstructs full parameter values. For both methods, we provide theoretical convergence results.
The numerical studies demonstrate that incorporating structural identifiability information can improve algorithm convergence and sampling efficiency. For both the SI and HIV models, the proposed methods achieve larger effective sample sizes and lower autocorrelation than standard random-walk MCMC. The improvements are particularly pronounced along structurally non-identifiable directions, where traditional MCMC methods struggle to fully explore the posterior density. 
The two approaches have complementary advantages. The geometric MCMC method operates in the original parameter space and can efficiently explore observationally equivalent manifolds through constrained geometric moves. The pseudo-marginal approach performs inference in a reduced-dimensional identifiable space, often resulting in superior mixing and convergence when efficient manifold sampling is available. These methods provide practical alternatives to standard MCMC algorithms for Bayesian inference in structurally non-identifiable models.

Several limitations of the proposed methods merit further investigation. First, both methods assume that structural identifiability information is available \emph{a priori}. Such information can often be obtained using differential algebraic approaches, for example with software such as DAISY \citep{J28}. However, these methods require substantial symbolic and algebraic computations, which limits their applicability to large-scale and highly complex models.
Second, structural identifiability represents the minimal level of non-identifiability under the assumption of sufficiently informative data. In practice, additional \emph{practical non-identifiability} may arise because of limited or noisy observations. Although the proposed methods improve sampling performance compared with conventional MCMC algorithms, they cannot fully resolve the challenges posed by practical non-identifiability. These challenges may be tackled by collecting more informative data or by developing Bayesian computation methods that explicitly account for practical non-identifiability.
From a computational perspective, the efficiency of the pseudo-marginal method depends on the ability to generate conditional samples on non-identifiable manifolds, whereas the geometric method requires repeated constrained projections, whose computational cost may increase with model complexity.
Future work can focus on extending the proposed framework to models with high-dimensional identifiable combination spaces, stochastic dynamical systems, and large-scale hierarchical Bayesian models.

This work shows that structural identifiability analysis can serve not only as a diagnostic tool for model calibration, but also as a principled foundation for the design of efficient Bayesian computation algorithms. By exploiting the geometric structure induced by identifiable combinations, it is possible to construct MCMC methods that remain both statistically valid and computationally efficient in settings where conventional sampling algorithms perform poorly. We hope that this perspective will encourage a closer integration of identifiability analysis and Bayesian computation in the development of reliable inference methods for increasingly complex mathematical models.







\bibliographystyle{plainnat} 

\bibliography{cas-refs}

\end{document}